\documentclass[runningheads]{llncs}
\usepackage[T1]{fontenc}
\usepackage{graphicx}
\usepackage{hyperref}
\usepackage{color}

\usepackage{amssymb}
\usepackage{amsmath}
\usepackage{algorithm}
\usepackage{algorithmic}
\usepackage{xcolor}
\usepackage{listings}
\usepackage{mdframed}
\usepackage[inline]{enumitem}
\usepackage{microtype}
\usepackage{tikz}
\usepackage{caption}
\usepackage{comment}
\usepackage{thmtools}
\usepackage{cleveref}
\usepackage{parskip}
\usetikzlibrary{arrows, shapes.geometric, automata}
\numberwithin{theorem}{section}

\AtBeginEnvironment{lemma}{\refstepcounter{theorem}}

\AtBeginEnvironment{proposition}{\refstepcounter{theorem}}

\AtBeginEnvironment{corollary}{\refstepcounter{theorem}}

\AtBeginEnvironment{definition}{\refstepcounter{theorem}}

\AtBeginEnvironment{example}{\refstepcounter{theorem}}
\begin{document}
\clearpage{}%
\newcommand{\lcm}{\textnormal{lcm}}
\newcommand{\abc}{\mathrm{\Sigma}}
\newcommand{\trans}{\delta}
\newcommand{\grid}{\mathcal{L}}
\newcommand{\neigh}{\mathcal{V}}
\newcommand{\qinit}{q_0} %
\newcommand{\app}{\mathord{:}}
\newcommand{\conc}{\cdot}
\newcommand{\powset}{\mathcal{P}}
\newcommand{\ultperiod}{\mathcal{T}}
\newcommand{\hierarchy}{\mathfrak{D}}
\newcommand{\A}{\mathrm{A}}
\newcommand{\B}{\mathrm{B}}
\newcommand{\C}{\mathrm{C}}
\newcommand{\ariths}{\mathrm{\Sigma}}
\newcommand{\arithp}{\mathrm{\Pi}}
\newcommand{\arithd}{\mathrm{\Delta}}

\newcommand{\N}{\mathbb{N}}

\newcommand{\fstred}{\triangleright^{\mathrm{F}}}
\newcommand{\fstredl}{\triangleleft^{\mathrm{F}}}
\newcommand{\fsteq}{\diamond^{\mathrm{F}}}
\newcommand{\cared}{\triangleright^{\mathrm{C}}}
\newcommand{\caredl}{\triangleleft^{\mathrm{C}}}
\newcommand{\caeq}{\diamond^{\mathrm{C}}}
\newcommand{\cacomp}{\cared_{\mathrm{comp}}}

\newcommand{\canred}{%
    \mathrel{\begin{tikzpicture}[baseline={(X.base)}]
        \node (X) {$\cared$};
        \draw (-0.18,-0.16) -- (-0.02,-0.16); %
        \draw (-0.16,-0.22) -- (-0.08,-0.11); %
    \end{tikzpicture}}%
}
\newcommand{\canredl}{%
    \mathrel{\begin{tikzpicture}[baseline={(X.base)}]
        \node (X) {$\caredl$};
        \draw (-0.22,-0.22) -- (-0.02,-0.22); %
        \draw (-0.16,-0.28) -- (-0.08,-0.16); %
    \end{tikzpicture}}%
}

\newcommand{\fstleq}{\leq^\mathrm{F}}
\newcommand{\fstgeq}{\geq^\mathrm{F}}
\newcommand{\fstlt}{<^\mathrm{F}}
\newcommand{\fstgt}{>^\mathrm{F}}
\newcommand{\fstdegree}{\equiv^\mathrm{F}}
\newcommand{\mmleq}{\leq^\mathrm{M}}
\newcommand{\mmgeq}{\geq^\mathrm{M}}
\newcommand{\mmlt}{<^\mathrm{M}}
\newcommand{\mmgt}{>^\mathrm{M}}
\newcommand{\mmdegree}{\equiv^\mathrm{M}}
\newcommand{\caleq}{\leq^\mathrm{C}}
\newcommand{\cageq}{\geq^\mathrm{C}}
\newcommand{\calt}{<^\mathrm{C}}
\newcommand{\cagt}{>^\mathrm{C}}
\newcommand{\cadegree}{\equiv^\mathrm{C}}

\newcommand{\tmred}{\triangleright^{\mathrm{TM}}}

\newcommand{\PD}{\mathsf{PD}}
\newcommand{\TM}{\mathsf{TM}}
\newcommand{\zip}{\mathsf{zip}}
\newcommand{\zeros}{\mathsf{zeros}}
\newcommand{\ones}{\mathsf{ones}}
\newcommand{\inv}{\mathsf{inv}}

\newcommand{\orbit}{\mathrm{\Delta}}

\clearpage{}%
\title{Cellular Automaton Reducibility as a Measure of Complexity for Infinite Words}
\titlerunning{CA Reducibility as a Measure of Complexity for Infinite Words}
\author{Markel Zubia\inst{1} \and %
Herman Geuvers\inst{2}}%
\authorrunning{M. Zubia and H. Geuvers}
\institute{Ruhr University Bochum, Germany, \email{markel.zubia@ruhr-uni-bochum.de} \and
Radboud University Nijmegen, The Netherlands, \email{herman@cs.ru.nl}}
\maketitle              %
\begin{abstract}
Infinite words, also known as streams, hold significant interest in computer
science and mathematics, raising the natural question of how their complexity
should be measured. We introduce \emph{cellular automaton reducibility} as a
measure of stream complexity: $\sigma$ is at least as complex as $\tau$ when
there exists a cellular automaton mapping $\sigma$ to $\tau$. This enables the
categorization of streams into degrees of complexity, analogous to Turing
degrees in computability theory. We investigate the algebraic properties of the
hierarchy that emerges from the partial ordering of degrees, showing that it is
not well-founded and not dense, that ultimately periodic streams are ordered by
divisibility of their period, that sparse streams are atoms, that maximal
streams have maximal subword complexity, and that suprema of sets of streams do
not generally exist. We also provide a pseudo-algorithm for classifying streams
up to this reducibility.
\end{abstract}

\section{Introduction}

The study of complexity is a central theme in various branches of theoretical
computer science and mathematics. Computational complexity theory, for instance,
quantifies the difficulty of decision problems and
algorithms~\cite{arora2009computational}. Similarly, computability theory
classifies problems based on their level of
(un)decidability~\cite{DBLP:books/daglib/0016363}. Within this broad landscape,
\emph{streams}, defined as infinite one-sided sequences over a finite alphabet,
play an important role. They arise naturally and are studied extensively in many
areas, including formal language theory~\cite{DBLP:books/automatic-sequences},
automata theory~\cite{pettorossi2022automata}, number
theory~\cite{berthe2010combinatorics,fernando2020p}, and dynamical
systems~\cite{hirsch2013differential,lind2021introduction}. Given their
ubiquity, it is relevant to ask about the \emph{complexity of streams}
themselves.

An approach for measuring stream complexity is to study transformations via
computational models. To be more specific, consider an automaton $M$ that takes
stream $\sigma$ as input and returns stream $\tau$ as output. Then, one can say
that $\sigma$ has enough information to reconstruct $\tau$; or, in other words,
$\sigma$ is at least as complex as $\tau$. Turing machines, while powerful,
present difficulties in this setting as all computable streams collapse into a
single equivalence class. Consequently, weaker models have been explored, such
as Mealy machines~\cite{rayna1974degrees,DBLP:journals/ita/Belovs08} and
finite-state transducers
(FSTs)~\cite{DBLP:journals/int/EndrullisHK11,DBLP:conf/cwords/EndrullisKSW15,endrullis2018degrees,DBLP:journals/logcom/Kaufmann23}.
These approaches motivate further investigations into stream transformations via
alternative models.

The objective of this paper is to establish a framework for comparing the complexity
of streams based on the computational model of cellular automata.
Cellular automata are discrete dynamical systems where spatially arranged cells
undergo updates over time. Each cell has content, usually referred to as the
state of the cell, which is updated at each time step according to a local rule
that depends solely on the states of the cell and of its neighboring
cells~\cite{wolfram1983statistical}. In this work, we introduce a measure of
stream complexity based on one-dimensional cellular automata~(1CA), by treating
streams as infinite one-sided arrays of cells.
This provides a mathematically natural definition for comparing streams, which
aligns with classical definitions of cellular
automata~\cite{wolfram1983statistical}. 

\begin{example} 
    Let $\sigma$ and $\tau$ be binary streams as shown below. Then, $\sigma$ can
    be transformed into $\tau$ by a 1CA with the following rule: a cell with
    state 0 becomes 1 if there is a 1 to its left or both its left and right
    neighbors have state 0; whereas any other cell transitions to state 0.
    Therefore, we say that $\sigma$ is at least as complex as $\tau$, or, more
    briefly, that $\sigma$ \emph{reduces to} $\tau$.

    \begin{center}
        \begin{tikzpicture}
        \node[align=left, inner sep=0pt] (equations) at (0,0) {$
            \begin{aligned}
            \sigma &= 10100100010000100000\dots\\~\\
            \tau &=   01010011001110011110\dots
            \end{aligned}
        $};

        \node[rotate=-90, scale=1.5] at (0.15,0.25) {$\}$};
        \node[rotate=-90] at (0.15,-0.1) {$\to$};
        \end{tikzpicture}
    \end{center}
\end{example}

The use of 1CAs for stream transformations offers two advantages over Mealy
machines and FSTs:
\begin{enumerate}
\item 1CAs induce a finer-grained hierarchy, which allows more nuanced
comparisons between streams. For example, while FSTs collapse all periodic
streams into a single degree, 1CAs are able to distinguish between them.
Similarly, the well-known period-doubling (PD) and Thue-Morse (TM) sequences,
defined formally in \cref{exa.PD,exa.TM}, are equivalent for FSTs, whereas we
conjecture that PD is strictly more complex than TM for 1CA reducibility.
\item The simplicity of 1CA reductions enables efficient algorithms for stream
classification (see \cref{sec.algorithm}), which in turn motivates its use in
practical applications. By contrast, FST synthesis is computationally
intractable~\cite{DBLP:conf/icalp/FiliotJLW16,DBLP:conf/mfcs/ExibardFJ18}.
\end{enumerate}

FSTs differ fundamentally from 1CAs in their way of transforming streams. An FST
performs a sequential left-to-right pass and uses a constant (albeit arbitrarily
large) amount of memory. In contrast, a 1CA applies a uniform local rule to all
cells globally, without a sequential aspect. Consequently, although this paper
engages with structures previously studied in the hierarchies induced by other
kinds of
automata~\cite{endrullis2018degrees,AleksandrsBelovs2007,rayna1974degrees}, our
treatment of these structures is inherently different. For example, work on FSTs
often relies on pumping-lemma-like arguments, whereas such reasoning is not
meaningful for 1CAs. Instead, we introduce novel techniques, including
diagonalization arguments (à la Cantor) and streams with carefully constructed
local neighborhoods.

\paragraph{Contribution}
We define cellular automaton transformations on streams based on the concept of
\emph{cellular automaton reducibility} in \cref{sec.CAreduction}, and we compare
it with FST reducibility in \cref{sec:fst_comparison}. In \cref{sec.degrees}, we
introduce the notion of \emph{degree} as the equivalence class induced by 1CA
reducibility, that is, a class of streams that can be reduced to one another. We
prove key structural properties of the hierarchy of stream degrees in the
subsequent \cref{sec.periodic,sec.sparse,sec:maxima}, including:
\begin{itemize} 
    \item it is not dense as there exist so-called ``atom'' degrees;
    \item it is not well-founded as infinite descending chains exist; 
    \item there exist infinite ascending chains; 
    \item ultimately periodic streams are ordered by divisibility; 
    \item sparse streams are either atoms or ultimately constant;
    \item maximal streams must have maximal subword complexity; 
    \item suprema of sets of streams do not generally exist. 
\end{itemize} 
\Cref{sec.algorithm} introduces an algorithm for checking if a stream reduces to
another, with proofs of correctness. In \cref{sec.alternatives}, we consider
alternative definitions and compare them with the main method. We start with
related work in \cref{sec.related}.

\section{Related Work}
\label{sec.related}

\paragraph{Kolmogorov Complexity}
Kolmogorov complexity measures the information content of a word as the size of
the smallest Turing machine that generates it~\cite{DBLP:series/txcs/LiV19}.
While extensively studied for infinite words~\cite{DBLP:journals/tcs/Staiger07},
it is undecidable even for finite words and difficult to approximate in
practice~\cite{zenil2020review}. One alternative is to classify streams by
mutual Turing transducibility, but this collapses all computable streams into a
single class, motivating investigations of reducibility via weaker computational
models.

\paragraph{Mealy Machines and Finite-State Transducers}
Mealy machines were among the first weaker models used to study stream
reducibility, initially in~\cite{rayna1974degrees} and later on
in~\cite{DBLP:journals/ita/Belovs08}. These automata output one letter per input
letter, thereby producing a corresponding output stream. Subsequent work shifted
to finite-state transducers (FSTs), which provide more flexibility by outputting
a finite word per input
letter~\cite{DBLP:journals/int/EndrullisHK11,DBLP:conf/cwords/EndrullisKSW15,endrullis2018degrees,DBLP:journals/logcom/Kaufmann23}.
These works study the algebraic properties of the hierarchies of streams induced
by reductions via Mealy machines and FSTs, paying close attention to extrema,
chains, antichains, density, and atom degrees (minimal nontrivial degrees).

\paragraph{Subword Complexity} 
Subword complexity measures a stream's complexity based on the diversity of its
finite substrings~\cite{DBLP:journals/tcs/EhrenfeuchtLR75,AleksandrsBelovs2007}.
Since this measure does not impose computational restrictions, it can assign
relatively low complexity measures to streams with high computational
complexity, including undecidable streams, which limits its utility from a
computational standpoint.

\paragraph{$\omega$-Languages, the Wadge Preorder, and the Wagner Hierarchy}
An \mbox{$\omega$-language} is a set of streams classified by the automata that
recognize them. The most prominent class is the \mbox{$\omega$-regular}
languages, recognized by Büchi automata~\cite{thomas1990automata}. Broader
classes include context-free~\cite{thomas1990automata,finkel2003borel} and
recursive $\omega$-languages~\cite{finkel2014ambiguity}. Closely related,
the Wadge preorder classifies \mbox{$\omega$-languages} via reductions through
continuous functions~\cite{van2012wadge}, and the Wagner hierarchy classifies
\mbox{$\omega$-regular} languages according to their position in the Wadge
hierarchy~\cite{DBLP:books/omega-languages}. Unlike these approaches, which
focus on language families, our work addresses the complexity of individual
streams. 

\paragraph{Computability- and Complexity-Theoretical Reductions}
Computability theory introduces several notions of reducibility that allow the
classification of decision problems, such as many-one reductions and Turing
reductions~\cite{shapiro1956degrees,DBLP:books/daglib/0016363}. Direct analogues
of these reductions are obtained by restricting to a polynomial runtime, and are
extensively studied in complexity theory~\cite{arora2009computational}. Our
notion of 1CA reducibility shares similarities with these classical reductions
but is specifically designed for stream comparison. 

\section{Preliminaries}
\label{sec.preliminaries}

We adopt Von Neumann ordinals, with $0 := \emptyset$ and $n + 1 := n \cup
\{n\}$. This means that $1 = \{0\}$, $2 = \{0, 1\}$, and so on. The first limit
ordinal is $\omega = \{0, 1, 2, \dots\}$. We write $n \mid m$ to denote that $n$
divides $m$. Congruence modulo $m$ is written as $a \equiv b \pmod{m}$. The
greatest common divisor and least common multiple of $n$ and $m$ are
respectively denoted by $\gcd(n, m)$ and $\text{lcm}(n, m)$.
Given sets $X$ and $Y$, a function with domain $X$ and codomain $Y$ is written
as $f : X \to Y$, and we denote the set of such functions by $Y^X$. Function
composition is given by $(g \circ f)(x) = g(f(x))$. A partial function $f : X
\rightharpoonup Y$ is a function $f : X \to Y \cup \{\bot\}$, where $\bot$
represents undefined values. We use underscores, ``$\_$'', as placeholders or
wildcards in function definitions; for instance, $f(\_, \_, x) = x$ defines the
function that ignores its first two arguments and returns the third.

\subsection{Finite and Infinite Words}

An alphabet is a finite set with at least two elements, and its members are
referred to as letters. Given an alphabet $\abc$, the set of words of length $n
\in \omega$ is $\abc^n$. A finite word $w \in \abc^n$ has length $|w| = n$. The
empty word is $\varepsilon$, with $\abc^0 = \{\varepsilon\}$. The set of all
finite words over $\abc$ is given by the Kleene closure $\abc^* = \bigcup_{n \in
\omega} \abc^n$. We also define $\abc^{\leq n} = \{w \in \abc^* \mid |w| \leq
n\}$. For $w \in \abc^*$ and $a \in \abc$, we denote word extension by $a\app
w$, meaning $a$ is prepended to $w$. Concatenation of words $w_1$ and $w_2$ is
denoted $w_1 \cdot w_2$, recursively defined by $\varepsilon \cdot w = w$ and
$(a\app w_1) \cdot w_2 = a\app (w_1 \cdot w_2)$. Repeated concatenation is
denoted by $w^n$, where $w^0 = \varepsilon$ and $w^{n+1} = w \cdot w^n$. For
binary words $w \in 2^{=n}$, we define the bitwise complement $\overline{w} \in
2^{=n}$ by $w(i) = 1 \Leftrightarrow \overline{w}(i) = 0$. We let $w(n:m) = w(n)
\cdot w(n+1) \cdots w(m)$. We will frequently refer to ``shifts'' of words: for
$w \in \abc^n$ and $i < n$, the $i$-shift is the cyclic permutation $w(i+1,
\dots, n-1) \cdot w(0, \dots, i)$.

An infinite stream over $\abc$ is a function $\sigma : \omega \to \abc$, also
seen as an element $\sigma \in \abc^\omega$. Streams are deconstructible as
$\sigma = a\app\sigma'$, where $a \in \abc$ is the head and $\sigma' \in
\abc^\omega$ is the tail. We use $2 = \{0, 1\}$ to denote the two-letter
alphabet, and the set of binary streams is $2^\omega$, known as the Cantor set.
Streams can also be constructed by concatenating infinitely many finite words,
denoted by $\prod_{i = 1}^\infty w_i$ for a sequence $\{w_i\}_{i \in \omega}$ of
words over $\abc$.

\subsection{Stream construction}\label{sec:stream_construction}

There are several methods for constructing streams, including automatic
sequences~\cite{DBLP:books/automatic-sequences}, morphic
words~\cite{DBLP:books/daglib/0025614}, Lindenmayer
systems~\cite{lindenmayer1968mathematical}, and Toeplitz
words~\cite{jacobs19690,cassaigne1997toeplitz}. We define streams coinductively
by the means of recurrent equations, as it is an expressive yet compact method
for constructing these.

\begin{example}\label{exa.TM} The \emph{Thue-Morse sequence}, with
code~[\href{https://oeis.org/A010060}{A010060}] in \emph{The On-Line
Encyclopedia of Integer Sequences} (OEIS)~\cite{oeis}, is defined
as $\text{TM} = 0 \app X$, where $X = 1\app\zip(X, Y)$, $Y = 0\app\zip(Y, X)$, and
$\zip(x\app\sigma, \tau) = x\app\zip(\tau, \sigma)$. Unfolding this
coinductive definition, 
\begin{align*}
\TM &= 0 \app X \\
    &= 0 \app 1 \app \zip(X, Y) \\
    &= 0 \app 1 \app \zip(1 \app \zip(X, Y), Y)\\
    &= 0 \app 1 \app 1 \app \zip(Y, \zip(X, Y))\\
    &= 0110100110010110\dots
\end{align*}

\begin{example}\label{exa.PD}
The \emph{period doubling sequence}, with
code~[\href{https://oeis.org/A096268}{A096268}] in OEIS~\cite{oeis}, is $\PD =
\zip(\ones, \inv(\PD))$, where $\ones = 1\app\ones$, and $\inv(x\app\sigma) =
1\app\inv(\sigma)$ if $x = 0$ and $0\app\inv(\sigma)$ otherwise. Unfolding this
definition yields $\PD = 1011101010111011\dots$
\end{example}

\end{example}

\subsection{Finite state transducers}\label{sec:prelim_fst}

Finite state transducers (FSTs) are finite automata with the ability to
transform one stream into another. They achieve this by processing each letter
from the input stream in a left-right pass, and producing a finite word as
output for each letter. For a detailed introduction to FSTs, we refer
to~\cite{DBLP:books/automatic-sequences}.

\begin{definition}
    A \emph{finite state transducer} (FST) is given as a tuple $(Q, \abc,
    \qinit, T, \lambda)$, where $Q$ is a finite set of \emph{states}, $\abc$ is
    a finite set with at least two elements called the \emph{alphabet}, $\qinit
    \in Q$ is the \emph{initial state}, $T : \abc \times Q \to Q$ is the
    \emph{transition function}, and $\lambda : \abc \times Q \to \abc^*$ is the
    \emph{output function}.
\end{definition}

\begin{definition}
    Given an FST $M = (Q, \abc, \qinit, T, \lambda)$, the \emph{stream
    transformation} $f_M : Q \times \abc^\omega \rightharpoonup \abc^\omega$
    corresponding to $M$ is 
    \[
    f_M(q, a\app\sigma) = \lambda(q, a) \conc f_M(T(q, a), \sigma),
    \] 
    and we use the shorthand notation $f_M(\sigma) = f_M(\qinit, \sigma)$.
\end{definition}

These transformations are modeled via partial functions, because FSTs can output
the empty word indefinitely, making the result finite and thus not a stream. 

FSTs provide a structured way to compare the complexity of
streams~\cite{DBLP:journals/int/EndrullisHK11,DBLP:journals/ita/Belovs08,rayna1974degrees}.
If a stream $\sigma$ can be transformed into another stream $\tau$ via an FST,
$\sigma$ is considered at least as complex as $\tau$. This leads to the formal
definition of FST reductions below.

\begin{definition}
    An \emph{FST reduction} is a relation $\fstred\subseteq \abc^\omega \times
    \abc^\omega$ where for all $\sigma, \tau \in \abc^\omega$, we have $\sigma
    \fstred \tau$ if and only if there exists an FST $M$ such that $f_M(\sigma)
    = \tau$. 
    
    Additionally, we let $\sigma \fstredl \tau \Leftrightarrow \tau \fstred
    \sigma$, and $\fsteq = \fstred \cap \fstredl$. 
\end{definition}

\begin{example}\label{ex:fst_zip} 
    As defined in \cref{sec:stream_construction}, the zip function interleaves two streams.
    \Cref{fig:fst_zip} shows an FST that reduces $\text{zip}(\sigma, \tau)$ to
    $\sigma$, which implies that $\text{zip}(\sigma, \tau) \fstred \sigma$. Each
    transition edge is labeled $a|w$, meaning letter $a$ is read and word $w$ is
    returned as output.
\end{example}

\begin{example}\label{ex:fst_pd_to_tm} \Cref{fig:fst_pd_to_tm} shows an FST that
    reads the period doubling sequence and outputs the Thue-Morse sequence,
    introduced in \cref{exa.TM,exa.PD}.
\end{example}

\begin{figure}[tb]
\begin{minipage}{0.45\textwidth}
    \centering
    \begin{tikzpicture}[->, >=latex,auto,node distance=3cm]
        \node[state] at (0, 0) (q0) {$q_0$};
        \node[state] at (2.5, 0) (q1) {$q_1$};

        \draw 
        (q0) edge[->, bend left] node[above] {$0|0$, $1|1$} (q1) 
        (q1) edge[->, bend left] node[below] {$0|\varepsilon$, $1|\varepsilon$} (q0) 
        ;

        \draw[->] (-1, 0) -- (q0);
    \end{tikzpicture}
    \captionof{figure}{An FST reducing $\text{zip}(\sigma, \tau)$ to $\sigma$ for any $\sigma, \tau \in \abc^\omega$.}
    \label{fig:fst_zip}
\end{minipage}
\hfill
\begin{minipage}{0.45\textwidth}
    \centering
    \begin{tikzpicture}[->, >=latex,auto,node distance=3cm]
        \node[state] at (0, 0) (q0) {$q_0$};
        \node[state] at (2.5, 1) (q1) {$q_1$};
        \node[state] at (2.5, -1) (q2) {$q_2$};

        \draw 
        (q0) edge[->] node[above] {$0|00$} (q1) 
        (q0) edge[->] node[below] {$1|01$} (q2) 
        (q1) edge[->, bend right] node[left] {$1|1$} (q2) 
        (q2) edge[->, bend right] node[right] {$1|0$} (q1) 
        (q1) edge[->, loop right] node[right] {$0|0$} (q1) 
        (q2) edge[->, loop right] node[right] {$0|1$} (q2) 
        ;

        \draw[->] (-1, 0) -- (q0);
    \end{tikzpicture}
    \captionof{figure}{An FST reducing $\PD$ to $\TM$.}
    \label{fig:fst_pd_to_tm}
\end{minipage}
\end{figure}

\section{Cellular Automaton Reducibility}
\label{sec.CAreduction}
This section introduces the notion of stream reducibility via 1D cellular
automata (1CA), and analyzes some of its basic algebraic properties. We
interpret streams as one-dimensional configurations of cells, which can be
updated by applying a local rule globally, yielding subsequent streams. The
definition introduced here aligns with the general notion of 1D cellular
automaton studied in classical work~\cite{wolfram1983statistical}, where the
local rule maps to a single letter. For alternative definitions of cellular
automaton reducibility, we refer to \cref{sec.alternatives}.

\begin{definition}
    A \emph{1D cellular automaton} (1CA) is a tuple $(\abc, N, \delta)$ where
    $\abc$ is a finite set with at least two elements called the
    \emph{alphabet}, $N \in \N$ is the \emph{neighborhood radius}, and
    $\delta : (\abc \cup \{\bot\})^{2N+1} \to \abc$ is the \emph{local rule}.
\end{definition}

\begin{definition}
    Given a stream $\sigma \in \abc^\omega$ and a 1CA $M = (\abc, N, \delta)$,
    the (local) \emph{neighborhood function} is the map 
    \[
    \neigh_{\sigma, N}(i) = (\sigma(i-N), \dots,\sigma(i),\dots,\sigma(i+N)),
    \] 
    where we let $\sigma(j) = \bot$ for any $j < 0$. The \emph{global update rule}
    of $M$ is a stream transformation $f_M : \abc^\omega \to \abc^\omega$ defined as 
    \[
    f_M(\sigma) = \delta \circ \neigh_{\sigma, N}.
    \]
\end{definition}

\begin{example}\label{ex:1ca_xor} 
    Fixing a binary alphabet, $\abc=2=\{0,1\}$, we can define a 1CA $M$ such that
    \[
    f_M(x\app y\app\sigma) = (x \oplus y) \app f_M(y\app\sigma),
    \] 
    where ``$\oplus$'' denotes addition modulo two (i.e.\ XOR). As each cell
    only needs to look at its immediate neighbor to the right, it is enough to
    let $N=1$ and define $\trans(\_, x, y) = x \oplus y$. Interestingly, this
    operation reduces the well-known sequences $\TM \cared
    \PD$~\cite{DBLP:journals/int/EndrullisHK11,DBLP:journals/delta-orbits}.
\end{example}

\begin{example}\label{ex:1ca_hello} 
    Given $\abc=\{a,b,c\}$, we can define a 1CA $M = (\abc, N, \delta)$ such
    that $f_M(\sigma) = ab \cdot \sigma$ by letting $N=2$ and defining 
    \[
        \delta(x, y, z, \_, \_) = \begin{cases}
            a \quad &\text{ if } x = y = \bot,\\
            b \quad &\text{ else if } x = \bot,\\
            x \quad &\text{ else.}
        \end{cases}
    \]
\end{example}

Through the following definition, we introduce the notion of 1CA reducibility as
the ability to transform one stream into another in one global update step.

\begin{definition}\label{def.cared}
    A \emph{1CA reduction} is a relation $\cared\subseteq \abc^\omega \times
    \abc^\omega$ where for all $\sigma, \tau \in \abc^\omega$, we have $\sigma
    \cared \tau$ if and only if there exists a 1CA $M$ such that $f_M(\sigma)
    = \tau$. 

    Additionally, we let $\sigma \caredl \tau \Leftrightarrow \tau \cared
    \sigma$, and $\caeq = \cared \cap \caredl$. 
\end{definition}

We show transitivity in the next proposition. Therefore, even though 1CA
reducibility represents a single update step, we can think of it as taking
finitely many steps.

\begin{proposition}\label{thm:preorder}
        Relation $\cared$ is a preorder.
\end{proposition}
\begin{proof}
    Let $\sigma, \tau, \rho \in \abc^\omega$.
    \begin{itemize}
        \item We get $\sigma \cared \sigma$ through the identity 1CA, $\trans(x) = x$. 
        \item Suppose $\sigma \cared \tau$ via $M_1 = (\abc, N_1, \trans_1)$ and
        $\tau \cared \rho$ via $M_2 = (\abc, N_2, \trans_2)$. Then, there exists
        $M_3 = (\abc, N_1 + N_2, \trans_3)$ such that $\sigma \cared \rho$ via
        $M_3$ by letting 
        \[
        \trans_3(x_0, \dots, x_{2(N_1 + N_2)}) := \trans_2(\trans_1(v_0), \dots,\trans(v_{2 N_2})),
        \]
        where $v_i := (x_{i}, \dots, x_{i+2N_1})$. \qed
    \end{itemize}
\end{proof}

\begin{corollary}
    $\caeq$ is an equivalence relation. \qed
\end{corollary}

\Cref{ex:1ca_hello} illustrates how a 1CA can be used to prepend the string $ab$
to a stream. This specific case motivates a more general result: 1CAs can insert
or remove finite strings, as stated in the following proposition.

\begin{proposition}\label{thm:append}
    For any $\sigma \in \abc^\omega$, $w \in \abc^*$, we have $w \cdot \sigma \caeq \sigma$.
\end{proposition}
\begin{proof}
    Given letter $a \in \abc$ and a stream $\sigma \in \abc^\omega$, we can
    append $a$ to $\sigma$ with $M = (\abc, 1, \delta)$ where 
    \[
    \delta(x, y, \_) = \begin{cases}
        a \quad &\text{if } x = \bot,\\
        x \quad &\text{else}.
    \end{cases}
    \]
    Since $\cared$ is transitive, 1CAs can insert or prepend finite words.

    Symmetrically, since the tail of a stream can be computed with the 1CA where
    $\delta(\_, \_, x) = x$, finitely many letters can be removed by
    transitivity.
    \qed
\end{proof}

\begin{corollary}\label{thm:mutations} 
    Relation $\cared$ is closed under finite mutations; in other words, for all
    $\sigma, \tau \in \abc^\omega$, if the set $\{i \in \N \mid \sigma(i) \neq
    \tau(i)\}$ is finite, then it holds that $\sigma \cared \tau$.
    \qed
\end{corollary}

Next, we show that any alphabet can be encoded in binary such that 1CA
reducibility is preserved. For this proof, it is crucial to take into account
that not every encoding suffices. To illustrate, consider naively encoding the
alphabet $\abc' = \{\A, \B, \C\}$ with $\psi(\A) = 00$, $\psi(\B) = 01$, and
$\psi(\C) = 10$. We overload $\psi(\sigma) = \prod_{i=0}^\infty \psi(\sigma(i))$
for any word $\sigma$. Then, the stream $\B^\omega$ is encoded as $(01)^\omega =
0101\dots$ and $\C^\omega$ is encoded as $(10)^\omega = 1010\dots$, which
eventually look indistinguishable from one another to any 1CA rule. While this
does not result in an immediate contradiction, as $B^\omega \caeq C^\omega$ and
also $(01)^\omega \caeq (10)^\omega$, it is easy to construct an example where
this naive encoding fails. For instance, if we let 
\begin{align*}
\sigma = \prod_{i=0}^\infty B^i C^i, \quad\quad\quad
\tau = \prod_{i=0}^\infty C^i B^i, 
\end{align*}
then $\sigma \caeq \tau$, whereas $\psi(\sigma)
\not \cared \psi(\tau)$, given that for any 1CA $(\abc, \trans, N)$, local
neighborhoods of $B^{N+1}$ and $C^{N+1}$ look indistinguishable, i.e.
\[
\neigh_{B^\omega, N}(N) = \neigh_{C^\omega, N}(N+2),
\]
yet, $\delta$ must map one to a 0 and the other to a 1. In the following result,
we construct a suitable reducibility-preserving encoding that avoids the
aforementioned issue.

\begin{lemma}\label{thm:binary_alphabet} 
    Let $\abc$ be an alphabet, that is, a finite set with at least two
    elements. Then, there exists a map $\phi : \abc \to 2^*$ such that
    \[
    \forall \sigma, \tau \in \abc^\omega, \sigma \cared \tau \Leftrightarrow \phi(\sigma) \cared \phi(\tau).
    \]
\end{lemma}
\begin{proof}
    To construct a suitable encoding, let $d = \lceil\log_2(|\abc|)\rceil$, and
    consider an injective $\gamma : \abc \to 2^d$. Then, the map 
    \[\phi(a) := 1 \cdot \gamma(a) \cdot 1 \cdot 0^{2^d + 1}\] satisfies the
    above bi-implication:

    ($\Rightarrow$) Suppose $\sigma \cared \tau$ with $M = (\abc, N, \trans)$,
    that is, for all $i \in \N$, $\delta(\neigh_{\sigma, N}(i)) = \tau(i)$. To
    construct a 1CA that reduces $\phi(\sigma) \cared \phi(\tau)$, we define
    $M' = (2, N', \trans')$ with $N' = N \cdot (2^{d+1} + 3)$. Notice that any
    substring of length $N'$ of $\phi(\sigma)$ is a shift of a word of the form
    \[
    1 \cdot w_1 \cdot 1 \cdot 0^{2^d+1} \cdot 
    1 \cdot w_2 \cdot 1 \cdot 0^{2^d+1} \cdots 
    1 \cdot w_N \cdot 1 \cdot 0^{2^d+1}.
    \]
    Refer to the shift of the above word as $w'$. Then, we obtain the desired
    reduction by defining $\delta'$ as follows: if the letter located in the
    middle of $w'$\dots
    \begin{itemize}
        \item \dots belongs to a patch of zeros of length $2^d + 1$, return 0.
        \item \dots is a 1 next to a patch of zeros of length $2^d + 1$, return 1.
        \item \dots belongs to one of the $w_i$. Since $\gamma$ is injective, we
        can consider the word $\prod_{j=1}^{N} \gamma^{-1}(w_j)$ and shift it
        such that $\gamma^{-1}(w_i)$ is centered; refer to it as $v$. Then, we
        return the letter in the middle of $\phi(\delta(v))$, shifted to align
        with $w'$, given that $\delta(v)$ yields the corresponding letter in
        $\tau$ due to $\sigma \cared \tau$ via $\delta$.
    \end{itemize}
    Therefore, by construction, we have
    \[
    f_{M'}(\phi(\sigma)) = 
    \prod_{j=0}^\infty 1 \cdot \gamma(\tau(j)) \cdot 1 \cdot 0^{2^d + 1} = \phi(\tau).
    \]

    ($\Leftarrow$) Suppose $\phi(\sigma) \cared \phi(\tau)$ with $M = (2, N,
    \delta)$, that is,
    \[
    \phi(\sigma) = \prod_{i=0}^\infty 1 \cdot \gamma(\sigma(i)) \cdot 1 \cdot 0^{2^d + 1}
    \]
    reduces to
    \[
    \phi(\tau) = \prod_{i=0}^\infty 1 \cdot \gamma(\tau(i)) \cdot 1 \cdot 0^{2^d + 1}.
    \]
    Arguing at a high-level due to the symmetry with the $(\Rightarrow)$
    direction, we let 1CA $M' = (2, N', \delta')$, where $N' = \lceil N /
    (2^{d+1} + 3) \rceil$, by extracting $\delta'$ from $\delta$. The idea is
    that if the word $w_0 \in 2^d$ in the neighborhood
    \[
    w = \prod_{i = -N'}^{N'} 1 \cdot \gamma(w_i) \cdot 1 \cdot 0^{2^d + 1}
    \]
    is mapped to some word $v \in 2^d$ (even though $\delta$ outputs individual
    bits, one can concatenate them to obtain $v$), we simply let
    \[
    \hspace{82pt}
    \delta'(\gamma^{-1}(w_{-N'}), \dots, \gamma^{-1}(w_{N'})) := \gamma^{-1}(v). 
    \hspace{82pt} \qed
    \]
\end{proof}

\Cref{thm:binary_alphabet} justifies restricting the alphabet to $2 = \{0, 1\}$,
as the algebraic structures in the hierarchy over any other alphabet will also
be present in the hierarchy over $2$. Therefore, we will mostly consider binary
streams in our examples, even though our theorems may still concern a general
alphabet denoted by $\abc$. It is worth noting that, while a similar lemma can
be proven for FSTs~\cite{DBLP:journals/int/EndrullisHK11}, methodologies vastly
differ: an FST processes a stream through a left-to-right pass, meaning that
the start of the stream can serve as a frame of reference, and thus a naive
encoding would suffice for showing that reducibility is preserved.

\subsection{Comparing FST reductions with 1CA reductions}\label{sec:fst_comparison}

This section provides a comparison between FST reducibility, as introduced in
\cref{sec:prelim_fst}, and 1CA reducibility. The key distinction lies in that
FSTs have memory, whereas 1CA do not. We exploit this difference to establish a
separation between the two families, in the spirit of classical separation
results~\cite{pettorossi2022automata}.

\begin{figure}
    \begin{center}
    \begin{tikzpicture}[->, >=latex,auto,node distance=3cm]
        \node[state] at (0, 0) (q0) {$q_0$};
        \node[state] at (3, -1.9) (q1) {$q_1$};
        \node[state] at (0, -1.9) (q2) {$q_2$};
        \node[state] at (3, 0) (q3) {$q_3$};

        \draw 
        (q0) edge[->, loop left] node[left] {$0|0$} (q0) 
        (q0) edge[->, bend left] node[right] {$1|1$} (q2) 
        (q2) edge[->, bend left] node[left] {$0|0$} (q0) 
        (q2) edge[->] node[above] {$1|0$} (q1) 
        (q1) edge[->, loop right] node[right] {$0|0$} (q1) 
        (q1) edge[->, bend left] node[left] {$1|1$} (q3) 
        (q3) edge[->] node[above] {$0|1$} (q0) 
        (q3) edge[->, bend left] node[right] {$1|1$} (q1) 
        ;
        ;

        \draw[->] (0, 1) -- (q0);
    \end{tikzpicture}
    \end{center}
    
    \caption{The transition diagram of an FST that reduces $\sigma$ to $\tau$ as
    defined in \cref{thm:fst_not_subseteq_ca}.}
    \label{fig:fst_not_ca}
\end{figure}
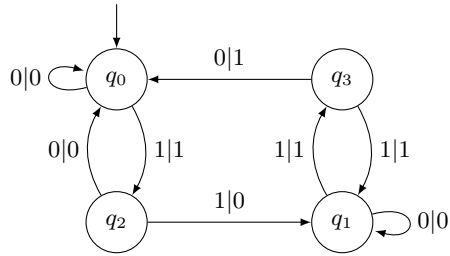

\begin{proposition}\label{thm:fst_not_subseteq_ca}
    $\fstred \not \subseteq \cared$.
\end{proposition}
\begin{proof}
    Let $\xi = (001)^\omega$, and define the following two streams:
    \[
    \sigma := \prod_{i \in \N} 0^{i + 1} \cdot 1 \cdot \xi(i+1), \quad \text{ and } \quad
    \tau := \prod_{i \in \N} 0^{i + 1} \cdot 1 \cdot \xi(i).
    \]
    It is possible to construct an FST that reduces $\sigma$ to $\tau$ as
    follows:
    \begin{enumerate}
        \item Skip all the padding 0's.
        \item When a 1 is reached, the next letter belongs to $\xi$. Remember it
        using a state, and replace it with the letter read previously.
        \item Go back to step 1.
    \end{enumerate}
    The explicit transition diagram of this FST is shown in \cref{fig:fst_not_ca}.
   
    Yet, a cellular automaton cannot perform this reduction. To see this,
    suppose that there is a 1CA $M = (\abc, N, \trans)$ where $f_M(\sigma) =
    \tau$. By construction of $\sigma$, there is a $j_0 \in \N$ where for all
    $j \geq j_0$, there exists $k \in \N$, where
    \[
    \neigh_{\sigma, N}(k) = 0^{N-1} \cdot 1 \cdot \xi(j+1) \cdot 0^N
    \]
    (component $\xi(j+1)$ being in the middle of the neighborhood). Take any
    such $j, j' \geq j_0$ satisfying $j \equiv 0 \mod 3$ and $j \equiv 2 \mod
    3$, and also take the corresponding $k, k' \in \N$. Since $\xi(j+1) =
    \xi(j'+1) = 0$, we get $\neigh_{\sigma, N}(k) = \neigh_{\sigma, N}(k')$.
    However, due to the assumption that $f_M(\sigma) = \tau$, we have
    \[
    0 = \xi(j) = \delta(\neigh_{\sigma, N}(k)) = \delta(\neigh_{\sigma,
    N}(k')) = \xi(j') = 1,
    \]
    a contradiction.
    \qed
\end{proof}

\begin{proposition}\label{thm:ca_subseteq_fst}
    $\cared \subseteq \fstred$.
\end{proposition}
\begin{proof}
    Suppose that $\sigma \cared \tau$ via the 1CA $M = (\abc, N, \trans)$. We
    construct an FST $M' = (Q, \abc, q_0, T, \lambda)$ for $\sigma \fstred \tau$
    by letting $Q = (\abc \cup \{\bot\})^{2N+1}$, $q_0 = (\bot, \dots, \bot)$,
    and defining $T$ and $\lambda$ as follows:
    \begin{enumerate}
        \item Start in $q_0 = (\bot, \dots, \bot)$, and, at any given time,
        record the next read letter in the current state and forget the earliest
        read letter (in a first-in-first-out manner):
        \[
        T(q, a) = (q(1), \dots, q(2N), a).
        \]
        \item When reading the first $N$ letters, the full neighborhood of the
        first letter in $\sigma$ still has not been read. Thus, we let
        $\lambda(q, a) = \varepsilon$, for any $q$ where $\bot$ appears in at
        least the first $N+1$ positions.
        \item Once the $N+1$-th letter is read, the entire neighborhood of the
        first cell is encoded in $q$, and we can thus define $\lambda(q, a) =
        \delta(q)$ for any $q$ with fewer than $N+1$ $\bot$'s.
    \end{enumerate}
    Therefore, it follows from the definition of $M'$ that
    \[
    f_{M'}(q_0, \sigma) = \prod_{i=0}^\infty \delta(\neigh_{\sigma, N}(i)) = f_M(\sigma) = \tau,
    \]
    and thus $\sigma \fstred \tau$.
    \qed
\end{proof}

Prior work introduces the notion of $\orbit$-orbits: the orbit obtained by
repeatedly applying the XOR
FST~\cite{DBLP:journals/delta-orbits,DBLP:journals/int/EndrullisHK11}. One can
check whether the orbits of two streams overlap, which can reveal interesting
connections between well-known streams that would otherwise be hard to discover.
Given that the XOR operation can be done via elementary cellular automata (as
constructed in \cref{ex:1ca_xor}), and a $\orbit$-orbit is nothing but
its space-time diagram, viewing this technique from the point of view of 1CA can
provide further insights, as we discuss in \cref{app:Orbits}.

\section{Degrees of Streams}
\label{sec.degrees}

In computability theory, sets of natural numbers are compared using Turing
machines. If two sets are Turing-equivalent, they belong to the same Turing
degree, and these degrees can then be compared with one another giving rise to a
hierarchy~\cite{shapiro1956degrees,DBLP:books/daglib/0016363}. This section
introduces definitions concerning cellular automaton reductions analogous to
those in computability theory. We start by formalizing the concept of
$\caeq$-degrees in the following definition.

\begin{definition}
    The \emph{hierarchy of $\caeq$-degrees} is the quotient $\hierarchy =
    \abc^\omega \slash \caeq$, and we refer to its members as
    \emph{$\caeq$-degrees}. The degree of $\sigma \in \abc^\omega$ is 
   denoted by $[\sigma]_{\caeq}$.
\end{definition}

Next, we extend the notion of 1CA reductions to compare degrees of streams.

\begin{definition}
    \emph{Stream reducibility} is a relation $\caleq\ \subseteq \hierarchy \times
    \hierarchy$ such that for all $X, Y \in \hierarchy$, we have
    $X \caleq Y \Leftrightarrow \forall \sigma \in X, \forall \tau \in Y, \ 
    \tau \cared \sigma$.
\end{definition}

The hierarchy forms a partially ordered set (poset) with $\leq^C$ by
\cref{thm:preorder}.

\begin{lemma}
    $(\hierarchy, \caleq)$ is a poset. \qed
\end{lemma}

Analogous to Turing degrees, $\caeq$-degrees are countable, given that every 1CA
has a finite description, which implies that there are countably many 1CAs.
Moreover, every $\caeq$-degree is infinite, as these degrees are closed under
finite mutations by \cref{thm:mutations}.

\begin{lemma}\label{thm:uncountable}
    Every $\caeq$-degree is countably infinite.
\end{lemma}
\begin{proof}
    Let $\sigma \in \abc^\omega$ and consider its degree $[\sigma]_{\caeq}$. For
    every $\tau \in [\sigma]_{\caeq}$, there exists a 1CA $M$ where $f_M(\sigma) =
    \tau$. Therefore, if we define $g(\tau) = M$, then $g$ is injective because
    $g(\tau) = g(\tau')$ implies 
    \[
    \tau = f_{g(\tau)}(\sigma) = f_{g(\tau')}(\sigma) = \tau'.
    \]
    Thus, as the set of 1CAs is countable, we have an injection from
    $[\sigma]_{\caeq}$ into $\N$.
    \qed
\end{proof}

\begin{corollary}
    There are $2^{\aleph_0}$ many $\caeq$-degrees. \qed
\end{corollary}

The hierarchy of $\caeq$-degrees of streams can be seen as a ``universe'' of
uncountably many degrees that can be compared via $\caleq$, each degree being
countably infinite, forming a poset structure.

\subsection{Ultimately $k$-Periodic Streams}\label{sec.periodic}

Ultimately periodic streams (those that eventually become periodic) are
interesting to examine due to their ubiquity and simplicity. Such streams belong
to the lowest degree in the hierarchies of Mealy machines and
FSTs~\cite{DBLP:journals/ita/Belovs08,DBLP:journals/int/EndrullisHK11}. This
section places these streams in the 1CA hierarchy. We start by briefly defining
stream periodicity, for completeness, which is a classical notion, see for
example~\cite{DBLP:books/automatic-sequences} for background on ultimately
periodic sequences and minimal periods.

\begin{definition}
  Let $\sigma \in \abc^\omega$. A number $p \in \N_{>0}$ is called a
  \emph{period} of $\sigma$ if there exists $n \in \N$ such that $\sigma(i+p) =
  \sigma(i)$ for all $i \ge n$. 
\end{definition}

\begin{definition}
  The \emph{set of periods} of $\sigma$ is denoted
  $\mathrm{Per}(\sigma)$. We say that $\sigma$ is \emph{ultimately periodic} if
  $\mathrm{Per}(\sigma) \neq \emptyset$.
  We call the minimal element in $\mathrm{Per}(\sigma)$ \emph{(minimal) period}
  of $\sigma$ and denote it by $\mathrm{per}(\sigma)$.
\end{definition}

\begin{definition}
  For $k \in \N^{>0}$, the set of \emph{ultimately $k$-periodic} streams is
  \[
    \ultperiod_k = \{\, \sigma \in \abc^\omega \mid \sigma \text{ ultimately periodic and } \mathrm{per}(\sigma)=k \,\}.
  \]
\end{definition}

\begin{definition}
  The set of \emph{ultimately periodic} streams is
  \[
    \ultperiod = \bigcup_{k \in \N_{>0}} \ultperiod_k.
  \]
\end{definition}

\begin{example}\label{ex:periodic}
  Consider the following binary streams:
  \[
    \sigma_2 = 1010101010\dots \quad \text{and} \quad
    \sigma_6 = 011101\,011101\,\dots.
  \]
  Then $\mathrm{per}(\sigma_2)=2$ and $\mathrm{per}(\sigma_6)=6$, hence
  $\sigma_2 \in \ultperiod_2$ and $\sigma_6 \in \ultperiod_6$.
  Although $6$ is also a period of $\sigma_2$, only the minimal period
  $\mathrm{per}(\sigma_2)$ matters for classification.
\end{example}

\begin{proposition}\label{thm:t1_low} 
    $\ultperiod_1$ is $\cared$-minimal for $\abc^\omega$, in other words,
    \[\forall \sigma \in \abc^\omega, \forall \kappa \in \ultperiod_1, \sigma
    \cared \kappa.\]
\end{proposition}
\begin{proof}
    Let $\sigma \in \abc^\omega$ and $\kappa = x \cdot y^\omega \in
    \ultperiod_1$ with $|y| = 1$. We have that $\sigma \cared y^\omega$ through
    $\delta(\_) = y$, and we have $y^\omega \cared \kappa$ given
    \cref{thm:mutations}. Thus, we get $\sigma \cared \kappa$ by transitivity.
    \qed
\end{proof}

Furthermore, since $\cared \subseteq \fstred$ by \cref{thm:ca_subseteq_fst} and
the fact that $\ultperiod$ is closed downwards under
$\fstred$~\cite{DBLP:journals/int/EndrullisHK11}, we have that $\ultperiod$ is
closed downwards under $\cared$.

\begin{corollary}\label{thm:closed_down} $\ultperiod$ is closed downward under
    $\cared$, that is, 
    \begin{align*}
    \hspace{86pt}   
    \forall \sigma \in \ultperiod, \tau \in \abc^\omega, \text{ if }
    \sigma \cared \tau \text{ then } \tau \in \ultperiod.
    \hspace{86pt} \qed
    \end{align*}
\end{corollary}

Even though we obtain \cref{thm:closed_down} through a fact about FSTs, it is
quite intuitive that applying the same local rule to all letters in a periodic
stream results in another periodic stream. Next, we prove helper
\cref{thm:closed_equiv,thm:period_red}, to later show in
\cref{thm:periodic_divisibility} that $\ultperiod_m$ reduces to $\ultperiod_n$
if and only if $n$ divides $m$.

\begin{lemma}\label{thm:closed_equiv}
    $\forall \tau, \tau' \in \ultperiod_k$, $\tau \caeq \tau'$.
\end{lemma}
\begin{proof}
    Let $\tau, \tau' \in \ultperiod_k$. We may assume that both $\tau$ and
    $\tau'$ are periodic, that is, $\tau = w^\omega$ and $\tau' = w'^\omega$,
    given that $\caeq$ is closed under finite mutations by \cref{thm:mutations}.
    Then, the map
    \[
        (w(i - k \mod k), \dots, w(i), \dots, w(i + k \mod k)) \mapsto w'(i)
    \]
    can be performed via a 1CA, because $\tau$ and $\tau'$ have no subperiods,
    i.e.\ $\tau, \tau' \not \in \ultperiod_{k'}$ for $k' < k$, thus local
    neighborhoods do not overlap with their own (nontrivial) shifts.
    \qed
\end{proof}

\begin{example}\label{ex:periodic_cared} Let's recall $\sigma_2 = 10101010\dots
    \in \ultperiod_2$ and $\sigma_6 = 011101\dots \in \ultperiod_6$ from
    \cref{ex:periodic}. Observe that the following rule with $N = 5$
    reduces $\sigma_6 \cared \sigma_2$:
    \begin{align*}
        \delta(\_, \_, \_, \_, \_, 0, 1, 1, 1, 0, 1) = 1, \quad\quad
        \delta(\_, \_, \_, \_, \_, 1, 1, 1, 0, 1, 0) = 0, \\
        \delta(\_, \_, \_, \_, \_, 1, 1, 0, 1, 0, 1) = 1, \quad\quad
        \delta(\_, \_, \_, \_, \_, 1, 0, 1, 0, 1, 1) = 0, \\
        \delta(\_, \_, \_, \_, \_, 0, 1, 0, 1, 1, 1) = 1, \quad\quad
        \delta(\_, \_, \_, \_, \_, 1, 0, 1, 1, 1, 0) = 0.
    \end{align*}
    However, the converse is not possible, that is, $\sigma_2 \not \cared
    \sigma_6$. This idea is formally captured in \cref{thm:period_red}
    and \cref{thm:periodic_divisibility} below.
\end{example}

\begin{lemma}\label{thm:period_red}
    For $\tau \in \ultperiod_m$, if $\tau \cared \sigma$, then $\sigma$ has an
    ultimate period of size $m$. 
\end{lemma}
\begin{proof}
    It is enough to show that the property holds when $\tau = (1 \cdot
    0^{m-1})^\omega$, because, for any other $\tau' \in \ultperiod_m$, if $\tau'
    \cared \sigma$, then $\tau \cared \sigma$ by \cref{thm:closed_equiv}.
    Thus, suppose $\tau \cared \sigma$ via $M = (2, N, \delta)$. Since for
    all $j \geq N + 1$, we have that $\neigh_{\tau, N}(j) = \neigh_{\tau, N}(j +
    m)$, we get
    \[
    \sigma(j) = \delta(\neigh_{\tau, N}(j)) =
    \delta(\neigh_{\tau, N}(j+m)) = \sigma(j+m),
    \] 
    in other words, $\sigma$ has an ultimate period of $m$. 
    \qed
\end{proof}

\begin{theorem}\label{thm:periodic_divisibility} $\ultperiod_n \caleq
    \ultperiod_m$ if and only if $n \mid m$.
\end{theorem}
\begin{proof}
    Recall that we use $n\mid m$ to express ``$n$ divides $m$''. Given that
    \cref{thm:period_red} implies that $\ultperiod_n \caleq \ultperiod_m
    \Rightarrow n \mid m$, we only need to show the $(\Leftarrow)$ direction.
    Suppose $n \mid m$, by \cref{thm:closed_equiv}, we only need to consider
    $\tau_m = (1 \cdot 0^{m-1})^\omega$ and $\tau_n = (1 \cdot 0^{n-1})^\omega$.
    Because $n \mid m$, we have $|w_m| = k \cdot |w_n|$ for some $k$. Thus, for $i =
    0, \dots, m - 1$, we can match $w_m(i)$ to $w_n(i \mod n)$ solely based on a
    finite neighborhood of size $m - 1$ similarly to the construction of the
    local rule in \cref{ex:periodic_cared}, since $w_m$ has no subperiods. 
    \qed
\end{proof}

\begin{corollary}\label{thm:closed} The following are implied by \cref{thm:periodic_divisibility}
  \begin{enumerate}
\item Every $\ultperiod_k$ is a $\caeq$-degree.
\item $\ultperiod_1$ is the lowest $\caeq$-degree.%
\item $(\{\ultperiod_k\}_{k \in \N}, \calt)$ is well-founded.%
\item  \label{thm:ascending_chain} 
    For all $k > 1$, the sequence $\ultperiod_1
    \calt \ultperiod_k \calt \ultperiod_{k^2} \calt \ultperiod_{k^3} \calt \dots$ is
    an infinite ascending chain. \qed
  \end{enumerate}
\end{corollary}
Classical literature defines that degree $P$ is an \emph{atom} (also known as a
\emph{prime} degree) when there is no degree strictly between $P$ and the bottom
degree $\ultperiod_1$. Prior work pays close attention to
atoms~\cite{DBLP:journals/int/EndrullisHK11,DBLP:journals/logcom/Kaufmann23}, as
their existence implies the hierarchy is not dense, and they can be used to
construct interesting streams such as extrema.

\begin{definition}\label{def:atom}
    Let $P \in \hierarchy$ and $P \neq \ultperiod_1$. $P$ is an \emph{atom} if
    and only if for all $X \in \hierarchy$, we have 
    $\ultperiod_1 \caleq X \caleq P \Rightarrow X = P \text{ or } X = \ultperiod_1$.
\end{definition}

\begin{corollary}\label{thm:primes}
    Degree $\ultperiod_p$ is an atom iff $p \in \N$ is a prime number. \qed
\end{corollary}

The existence of an infinite descending chain in $\hierarchy$ would imply that
the hierarchy is not well-founded. While ultimately $k$-periodic streams do not
form such chains, we can use them to construct one. The idea is to encode the
information from each $\ultperiod_k$ into a single stream, which we can
abstractly represent as $\langle \ultperiod_2, \ultperiod_3, \ultperiod_4,
\ultperiod_5, \dots \rangle$. If successful, we can strictly reduce this to
$\langle \ultperiod_1, \ultperiod_3, \ultperiod_4, \ultperiod_5, \dots \rangle$,
to then reduce it to $\langle \ultperiod_1, \ultperiod_1, \ultperiod_4,
\ultperiod_5, \dots \rangle$, and so on. Although this tuple is not itself a
stream, we can get the same effect by concatenating increasingly long finite
prefixes. Let $\tau_k = (1 \cdot 0^{k-1})^\omega \in \ultperiod_k$. Define the
first stream in our descending chain as 
\[
\mu_1 = \tau_2^{\leq 1} \cdot \tau_2^{\leq 2} \cdot \tau_3^{\leq 3} \cdot 
\tau_2^{\leq 4} \cdot \tau_3^{\leq5} \cdot \tau_4^{\leq 6} \cdot 
\tau_2^{\leq 7} \cdot \tau_3^{\leq 8} \cdot\tau_4^{\leq 9} \cdot \tau_5^{\leq 10} \cdot \dots,
\]
 where $\tau_k^{\leq l}$
denotes the prefix of $\tau_k$ of length $l$. Then, we obtain an infinite
descending chain  as follows:

\begin{proposition} \label{thm:antichain} 
    There is an infinite descending chain
    \begin{align*}
        &\mu_1 = \tau_2^{\leq 1} \cdot \tau_2^{\leq 2} \cdot \tau_3^{\leq 3} 
        \cdot \tau_2^{\leq 4} \cdot \tau_3^{\leq 5} \cdot \tau_4^{\leq 6} 
        \cdot \tau_2^{\leq 7} \cdot \tau_3^{\leq 8} \dots\\
        \canred
        &\mu_2 = \tau_1^{\leq 1} \cdot \tau_1^{\leq 2} \cdot \tau_3^{\leq 3} 
        \cdot \tau_1^{\leq 4} \cdot \tau_3^{\leq 5} \cdot \tau_4^{\leq 6} 
        \cdot \tau_1^{\leq 7} \cdot \tau_3^{\leq 8} \dots\\
        \canred
        &\mu_3 = \tau_1^{\leq 1} \cdot \tau_1^{\leq 2} \cdot \tau_1^{\leq 3} 
        \cdot \tau_1^{\leq 4} \cdot \tau_1^{\leq 5} \cdot \tau_4^{\leq 6} 
        \cdot \tau_1^{\leq 7} \cdot \tau_1^{\leq 8} \dots\\
        \canred &\dots,
    \end{align*}
    where ``$\canred$'' denotes strict reducibility.
\end{proposition}
\begin{proof}
    For all $i \in \N$, we need to show that $\mu_i \canred \mu_{i + 1}$. We
    start by showing that $\mu_{i + 1} \not \cared \mu_i$, by assuming that
    $\mu_{i + 1} \cared \mu_i$ via $M = (\abc, N, \delta)$, for a contradiction.
    By construction of $\mu_i$ and $\mu_{i+1}$, there are increasingly long
    substrings of $\tau_1$ in $\mu_{i+1}$ that get mapped to equally long
    substrings of $\tau_{i+1}$ in $\mu_i$. Given that said substrings exceed the
    neighborhood size, $2N + 1$, we have that $f_M(\tau_1) = \tau_{i+1}$,
    contradicting \cref{thm:periodic_divisibility}.
    
    We show $\mu_i \cared \mu_{i+1}$ by providing a suitable 1CA $M = (\abc, N,
    \delta)$. Rule $\delta$ must identify whether a 0 belongs to a substring of
    $\tau_{i+1}$, and, only if so, turn it into a 1. We define such a rule as
    follows:
    \begin{itemize}
        \item If the letter in the center is a 1, return 1.
        \item If it is a zero in a patch of $k \leq i$ zeros, return 1 if and only if
        the neighborhood has either form $\dots 1\cdot1\cdot0^k \cdot 1 \dots$
        or $\dots1 \cdot 0^i \cdot 1 \cdot 0^k \cdot 1 \dots$.
        \item In any other case, return 0.
    \end{itemize}
    This operation can be performed by a rule with $N = 2i + 1$. Notice that,
    since $\mu_i$ consists of increasingly long periodic substrings, $M$
    ``eventually correctly reduces $\mu_{i}$ to $\mu_{i+1}$'', that is, 
    \[
    \exists j \in \N, \forall j' \geq j, f_M(\mu_i)(j') = \mu_{i+1}(j').
    \] 
    Thus, $f_M(\mu_i) \caeq \mu_{i+1}$ by \cref{thm:mutations}, which means that
    $\mu_i \cared \mu_{i+1}$.
    \qed
\end{proof}

\begin{example}
    Stream $\mu_3$, as in \cref{thm:antichain}, can be reduced to $\mu_2$ via
    the 1CA with $N = 5$ with the following rule:
    \[
    \begin{tabular}{l cccccccccccc l}
    $\delta($ & \_, & \_, & \_, & \_, & \_, & 1, & \_, & \_, & \_, & \_, & \_, & \_) & $= 1$,\\
    $\delta($ & \_, & \_, & \_, & 1,  & 1,  & 0, & 1, & \_, & \_, & \_, & \_, & \_) & $= 1$,\\
    $\delta($ & \_, & \_, & \_, & 1,  & 1,  & 0, & 0, & 1, & \_, & \_, & \_, & \_) & $= 1$,\\
    $\delta($ & \_, & \_, & 1,  & 1,  & 0,  & 0, & 1, & \_, & \_, & \_, & \_, & \_) & $= 1$,\\
    $\delta($ & \_, & 1,  & 0,  & 0,  & 1,  & 0, & 1, & \_, & \_, & \_, & \_, & \_) & $= 1$,\\
    $\delta($ & \_, & 1,  & 0,  & 0,  & 1,  & 0, & 0, & 1, & \_, & \_, & \_, & \_) & $= 1$,\\
    $\delta($ & 1,  & 0,  & 0,  & 1,  & 0,  & 0, & 1, & \_, & \_, & \_, & \_, & \_) & $= 1$,
    \end{tabular}
    \]
    and $\delta$ returns $0$ for any other input.
\end{example}

\subsection{Sparse Streams}\label{sec.sparse}

Cantor's diagonal argument can be used to construct an undecidable set $K
\subseteq \N$; for every Turing machine $\varphi_i$, let $i \in K$ if
$\varphi_i(i)$ halts, and $i \not \in K$ otherwise. A similar approach can be
used to construct Turing-incomparable sets~\cite{kleene1954upper,friedberg1957two,muchnik1956unsolvability}, which serves
as our inspiration for constructing aperiodic $\cared$-incomparable streams.

Notice that the string $00$ cannot be mapped to the string $01$ by a 1CA of
radius 0, because the first 0 is mapped to a 0, while the second 0 is mapped to
a 1. Therefore, if we let $\sigma_0 = 001$ and $\tau_0 = 010$, the two strings
are incomparable via any 1CA with radius 0. Similarly, by letting $\sigma_1 =
000\ 000\ 010$ and $\tau_1 = 000\ 010\ 000$, there is no 1CA with $N = 1$ that
can compare $\sigma$ and $\tau$. It is easy to see the general pattern for
$\sigma_i$ and $\tau_i$. Thus, the streams $\sigma$ and $\tau$ obtained by
concatenating all $\sigma_i$ and $\tau_i$ respectively are
$\cared$-incomparable.

Observe that the incomparable streams just constructed are mostly composed of
0's. Intuition tells us that, in general, streams containing many 0's must
``diagonalize'' against the neighborhood radius, making them hard to
compare through 1CAs. This motivates the introduction of the notion of stream
sparsity.

\begin{definition}
    Stream $\sigma \in \abc^\omega$ is \emph{sparse} iff there exists a padding
    letter $p \in \abc$ where $\forall d \in \N, \exists n_0 \in \N,
    \forall i \geq n_0,$
    $
        \sigma(i) \neq p \Rightarrow \sigma(i + j) = p, 
        \text{ for } j = 1, \dots, d.
    $
\end{definition}

In other words, a stream is sparse if for every $d \in \N$ there exists $n_0
\in \N$ such that beyond $n_0$, each non-$p$ letter is followed by at least $d$
consecutive $p$'s.

\begin{example}
    The following stream is sparse with $p = 0$:
    \[\prod_{i = 0}^\infty 1 \cdot 0^i = 110100100010000\dots\]
\end{example}

For simplicity, we restrict our attention to $\abc = 2$ (i.e., $\{0,1\}$), and,
if we consider a sparse stream, we assume $0$ is the padding letter.
Nevertheless, the results in this section directly carry over to arbitrary
sparse streams.

It is hard to compare two sparse streams through 1CAs, since eventually almost
the entire neighborhood will consist of only the padding letter $0$. In fact, we
next show that sparse streams that are not eventually constant are atoms.

\begin{theorem}\label{thm:sparse} 
    The $\caeq$-degree of any sparse stream is either an atom or $\ultperiod_1$.
\end{theorem}
\begin{proof}
    Let $\zeta \in 2^\omega$ be sparse. If $\zeta \in \ultperiod_1$ we are done.
    Thus, let $\zeta \not \in \ultperiod_1$, that is, let it have the form
    \[
    \zeta = \prod_{i = 0}^\infty 1 \cdot 0^{c(i)},
    \]
    for some monotone unbounded $c : \N \to \N$. Now, suppose that $\zeta
    \cared \sigma$ via $M = (\abc, N, \delta)$. Without loss of generality, we
    may assume that $\delta(0, \dots, 0) = 0$. Eventually, every neighborhood
    will contain at most one $1$, so we can focus on these, and deal with the
    finite discrepancies via \cref{thm:mutations}. Define
    \[
    w := \prod_{i = 0}^{2N} \delta(
       0^{2N - i} \cdot 1 \cdot 0^{i});
    \]
    then, given that $\zeta \cared \sigma$ and $\zeta$ is sparse, we have that
    $\sigma$ eventually takes the form ``$w$ followed by increasingly many
    zeros, repeating''. More formally, the following stream
    \[
    \rho := \prod_{i=0}^\infty w \cdot 0^{c(i) - |w|},
    \]
    where we let $0^n = \varepsilon$ for $n < 0$, has the property that from
    some point $j \in \N$ onward, $\sigma(j) = \rho(j)$, and thus $\sigma \caeq
    \rho$ by \cref{thm:mutations}.
    There are two cases:
    \begin{enumerate}
        \item $w = 0^{2N + 1}$, in which case we get $\rho = 0^\omega$, meaning
        that $\sigma \in \ultperiod_1$ and we are done.
        \item $w \neq 0^{2N + 1}$, in which case we can retrieve $\zeta$ back
        from $\rho$ (except for finitely many letters) via the rule
        \[
        \delta(x) = 
        \begin{cases}
            1 \quad \text{if } x = w,\\
            0 \quad \text{else}.
        \end{cases}
        \]
        This implies that $\sigma \caeq \zeta$.
    \end{enumerate}

    We have shown $\forall \sigma \in 2^\omega$, $\zeta \cared
    \sigma$ $\Rightarrow$ $\sigma \in \ultperiod_1$ or $\sigma \caeq \zeta$,
    which means that $[\zeta]_{\caeq}$ is an atom by \cref{def:atom}. 
    \qed
\end{proof}

\begin{corollary}
    There are $2^{\aleph_0}$ many atom $\caeq$-degrees.
\end{corollary}
\begin{proof}
    There are $2^{\aleph_0}$ many sparse streams, meaning that there are
    $2^{\aleph_0}$ distinct degrees of sparse streams, as each degree is
    countable by \cref{thm:uncountable}, and \cref{thm:sparse} shows they are
    all atoms. 
    \qed
\end{proof}

A natural next step is to determine the streams that sparse streams reduce to.
To address this, we introduce the notion of stream congruence.

\begin{definition}
  Streams $\sigma, \tau \!\in 2^\omega$ are \emph{congruent}\footnote{Note that
  this notion of congruence should be extended up to letter renaming (a
  bijection $\phi : \abc \to \abc$) in case of arbitrary alphabets for the proofs
  to carry over.} whenever there are
    $n, m \! \in \! \N$ such
    that $\forall i \in \N$, $\sigma(n + i) = \tau(m + i)$.
    We denote congruence as $\sigma \cong \tau$.
\end{definition}

It is clear that $\sigma \cong \tau$ implies $\sigma \caeq \tau$
(\cref{thm:mutations}).

\begin{proposition}\label{thm:congruence} For every sparse $\sigma, \tau \in \abc^\omega$,
    $\sigma \caeq \tau$ iff $\sigma \cong \tau$.
\end{proposition}
\begin{proof} 
It is enough to show the ($\Rightarrow$) direction. Let $\sigma \cared \tau$
through the 1CA $M \! = \! (\abc, N, \delta)$. Then, the following two hold
\begin{itemize}
    \item $\delta(0, \dots, 0) = 0$.
    \item Out of all neighborhoods containing exactly
    one $1$, exactly one of them must output $1$, whereas all the
    others output $0$.
\end{itemize}
Therefore, after some initial part where $\tau(i)$ is determined by the value of
$\delta$ on neighborhoods containing a $\bot$ or multiple $1$'s, we find copies
of $1$ in $\tau$ at a fixed distance from the corresponding $1$ in $\sigma$. 
Thus, we obtain the one-to-one mapping $\phi : \abc \to \abc$ where
$\phi(p_1) = p_2$ and $\phi(x) = y$. Moreover, $\phi^{-1}$ is
given through symmetrical reasoning given $\tau \cared \sigma$,
so $\phi$ is a bijection and $\sigma \cong \tau$. 
\qed
\end{proof}

Sparsity, as defined in this paper, is a quite restrictive property compared to
weaker notions of sparsity introduced in the
literature~\cite{DBLP:books/automatic-sequences}. For instance, streams such as
$\prod_{i \in \N} 11 \cdot 0^i$ are considered sparse in alternative settings,
whereas it is not sparse in the sense studied so far here. Motivated by this, we
next introduce and study a weaker notion.

\begin{definition}
    A stream $\sigma \in \abc^\omega$ is \emph{weakly sparse} if and only if
    there is a padding letter $p \in \abc$ such that 
    \[
        \lim_{n \to \infty} \frac{|\{i \leq n \mid \sigma(i) = p\}|}{n} = 1.
    \]
\end{definition}

\begin{example}\label{ex:almost_sparse}
    For any finite word $w \in 2^*$, the stream 
        $\prod_{i \in \N} w \cdot 0^i$
    is weakly sparse, but it is not sparse unless $w$ has at most one 1.
\end{example}

\begin{example}\label{ex:prime_gaps}
    The \emph{stream of primes} is $\pi \in 2^\omega$ where $\pi(i) = 1$ if and
    only if $i$ is a prime number. Stream $\pi$ is weakly sparse due to the
    prime number theorem, while Zhang's theorem states that there are infinitely
    many prime gaps of size $k$ for some constant $k \in
    \N$~\cite{zhang2014bounded}, which implies that $\pi$ is not sparse.
\end{example}

We will again restrict to $\abc = 2 = \{0,1\}$, and assume that $0$ is the
padding letter. Next, we prove that being sparse is stricter than being weakly
sparse, that is, every sparse stream is weakly sparse, while the converse does
not hold as was shown in \cref{ex:almost_sparse,ex:prime_gaps}.

\begin{lemma}
    $\forall \sigma \in \abc^\omega$, if $\sigma$ is sparse then $\sigma$ is also
    weakly sparse.
\end{lemma}
\begin{proof}
    Let stream $\sigma$ be sparse. Unfolding the limit in the definition of weak
    sparsity, the property that we need to prove can be reformulated as follows: 
    \[ \forall \gamma \in
    \mathbb{R}^{>0}, \exists n_0 \in \N, \forall n \geq n_0,\ \frac{|\{i \leq n
    \mid \sigma(i) = p\}|}{n} \geq 1 - \gamma. 
    \] 
    This reformulation has clear overlap with the definition of sparsity. 
    To get an explicit proof, let $\gamma > 0$ and take any $d \in \N$. Then,
    by sparsity, there exists $N \in \N$ such that for all $i \ge N$, if
    $\sigma(i) \neq p$ then $\sigma(i+1),\dots,\sigma(i+d)$ are all equal to
    $p$. Therefore, among any $d{+}1$ consecutive positions after $N$, at most
    one is not $p$.
    
    Hence for all $n \ge N$,
    \[
        |\{N < i \le n \mid \sigma(i) \neq p\}| \le \frac{n-N}{d+1},
    \]
    and thus $|\{i \le n \mid \sigma(i) \neq p\}| \le \frac{n-N}{d+1} + N$.
    Dividing by $n$ yields
    \[
        \frac{|\{i \le n \mid \sigma(i) \neq p\}|}{n}
        \le \frac{1}{d+1} + \frac{N}{n}.
    \]
    Taking complements,
    \[
        \frac{|\{i \le n \mid \sigma(i)=p\}|}{n}
        \ge 1 - \frac{1}{d+1} - \frac{N}{n}.
    \]
    To conclude the proof, choose $d$ such that $\frac{1}{d+1} \leq
    \frac{\gamma}{2}$, and also choose $n_0 \ge N$ such that $\tfrac{N}{n} \le
    \frac{\gamma}{2}$ for all $n \ge n_0$. 
    \qed
\end{proof}

\begin{lemma}
    If $\sigma$ is sparse, then every $\tau \in [\sigma]_{\caeq}$ is weakly
    sparse.
\end{lemma}
\begin{proof}
    Recall cases 1.\ and 2.\ from \cref{thm:sparse}. In case 1.\ we have
    $\sigma, \tau \in \ultperiod_1$, so it is weakly sparse. In case 2.\
    excluding finitely many letters, $\tau$ eventually has the form
    \[
    \prod_{i=0}^\infty w \cdot 0^{c(i) - |w|},
    \]
    with $0^n = \varepsilon$ if $n < 0$. Recall that $c : \N \to \N$ is
    monotone unbounded, which implies that $\tau$ is weakly sparse.
    \qed
\end{proof}

Given a stream $\sigma \in \abc^\omega$, we can construct a weakly sparse stream
containing the information from $\sigma$ by interleaving it with increasingly
many padding letters. The following result uses this principle to encode the
well-founded structure of periodic streams into a family of weakly sparse
streams.

\begin{proposition}
    Define the following weakly sparse stream:
    \[
        \varsigma_i = \prod_{j = 0}^\infty (1 \cdot 0^{i-1})^j \cdot 0^{2^j},
    \]
    Then, for any $n, m \in \N$, we have $\varsigma_m \cared
    \varsigma_n$ if and only if $n \mid m$.
\end{proposition}
\begin{proof}
    It is clear that $\varsigma_i$ is weakly sparse, as $O(2^j) \gg O(j \cdot
    i)$ for constant $i$. We therefore focus on proving the bi-implication.

    $(\Rightarrow)$ Suppose that $\varsigma_m \cared \varsigma_n$ with machine
    $M = (\abc, N, \delta)$. By construction of $\varsigma_m$ and $\varsigma_n$,
    arbitrarily large substrings (even larger than the local neighborhood of
    $M$) of $(1 \cdot 0^{m - 1})^\omega \in \ultperiod_m$ appear in
    $\varsigma_m$, which are mapped to substrings of $(1 \cdot 0^{n - 1}) \in
    \ultperiod_n$ in $\varsigma_n$ by $M$ due to $f_M(\varsigma_m) =
    \varsigma_n$. Therefore, we have that $f_m(\tau_m) \caeq \tau_n$, which
    implies that $n \mid m$ by \cref{thm:periodic_divisibility}.

    $(\Leftarrow)$ Suppose that $n \mid m$, i.e.\ $\exists k \in \N, m = n
    \cdot k$. Then, we construct a 1CA $M = (\abc, N, \delta)$ to reduce
    $\varsigma_m \cared \varsigma_n$ by mapping every $1 \cdot 0^m$ in the input
    stream to $(1 \cdot 0^n)^k$ in the output.
    \qed
\end{proof}

\subsection{Maximal degrees and suprema}\label{sec:maxima}

So far, we have established the existence of a minimal (bottom) degree. A
natural next question is whether maximal degrees exist. This question is
notoriously difficult: it remains open for finite-state
transducers~\cite{DBLP:journals/int/EndrullisHK11}, and it also remains
unresolved in our setting. Nevertheless, we can obtain partial results in this
direction, beginning with the following theorem:

\begin{theorem}\label{thm:maximal}
    For any $\sigma \in \abc^\omega$, if there exists a finite word $v \in
    \abc^*$ that appears only finitely many times in $\sigma$, then $\sigma$ is not
    maximal.
\end{theorem}
\begin{proof}
    Suppose $v \in \abc^*$ occurs only finitely many times in $\sigma$.
    Construct a word $\hat v$ containing $v$ that does not overlap
    (nontrivially) with its own shifts. For example, one may take $\hat v := 1
    \cdot v \cdot 1 \cdot 0^{|v|+1}$. Since $v$ occurs only finitely often in
    $\sigma$, so does $\hat v$. By the pigeonhole principle, there exists some
    word $w \in \abc^*$ of length $|\hat v|$ that occurs infinitely often in
    $\sigma$. Consider the process of modifying $\sigma$ by replacing an
    arbitrary subset of these infinitely many occurrences of $w$ with $\hat v$.
    Each such modification yields a new infinite stream $\tau$. Two facts are
    immediate: 1.~Every such $\tau$ can be reduced back to $\sigma$ by replacing
    all $\hat v$'s by $w$ (as $\hat v$ does not overlap its own shifts), and
    then compensating for the finitely many original $\hat v$'s that were
    already in $\sigma$ via a finite mutation. 2.~There are $2^{\aleph_0}$
    distinct ways to choose which copies of $w$ to replace. Hence, $\sigma$
    cannot reduce to all of these streams, which means there exist streams
    strictly above $\sigma$. Therefore, $\sigma$ is not maximal. 
    \qed
\end{proof}

By \cref{thm:maximal}, a maximal stream must contain every finite subword
infinitely often, i.e., it must have maximal subword
complexity~\cite{DBLP:journals/tcs/EhrenfeuchtLR75}. Let $\mathcal{S} \subseteq
\abc^\omega$ be the set of such streams. We now show that not all streams in
$\mathcal{S}$ are maximal for 1CA reducibility.

\begin{lemma}\label{thm:non_maximal}
    There exist non-maximal streams in $\mathcal{S}$.
\end{lemma}

\begin{proof}
    Take $\sigma \in \mathcal{S}$, and let $M$ be the XOR 1CA from
    \cref{ex:1ca_xor}. Letting $f_M(\sigma) = \tau$, we claim that $\tau \in
    \mathcal{S}$, yet $f_M(\sigma)$ is not maximal.

    \begin{itemize}
        \item 
        Due to surjectivity, every finite word appears infinitely often in
        $f_M(\sigma)$: for instance, if one asks whether $011$ occurs infinitely
        often in $f_M(\sigma)$, the answer is yes because its preimage $1101$
        occurs infinitely often in $\sigma$. Hence $f_M(\sigma) \in
        \mathcal{S}$. 
        \item  
        On the other hand, due to the non-injectivity of the XOR automaton,
        $f_M(\sigma)$ cannot be reduced back to $\sigma$. To show this, suppose
        for a contradiction that $f_M(\sigma) \cared \sigma$ via $M' = (2, N,
        \delta)$. For any word $w \in 2^{2N+1}$, both words $0 \cdot w$ and $1
        \cdot w$ must appear in $f_M(\sigma)$, given that $f_M(\sigma) \in
        \mathcal{S}$. Concretely, there exist indices $i$ and $j$ in
        $f_M(\sigma)$ where 
        \[
        \neigh_{f_M(\sigma), N}(i) = w = \neigh_{f_M(\sigma), N}(j),
        \]
        as well as having the letter to the left of $\neigh_{f_M(\sigma), N}(i)$
        be a 0 (i.e.\ $f_M(\sigma)(i-N-1) = 0$) and to the left of
        $\neigh_{f_M(\sigma), N}(j)$ be a 1 (i.e.\ $f_M(\sigma)(j-N-1) = 1$).
        Then, since there is a 0 to the left of $w$ in one case, and a 1 in the
        other case, the parity of the output of $M'$ is flipped (as it undoes
        the XOR operation). Thus,
        \[
        \delta(\neigh_{f_M(\sigma), N}(i)) \neq
        \delta(\neigh_{f_M(\sigma), N}(j)),
        \] a contradiction. \qed
    \end{itemize}
   \end{proof}

The hierarchies of Mealy machines and FSTs form semi-lattices, with the zip
operator (see \cref{ex:fst_zip}) providing the least upper bound of two
streams~\cite{DBLP:journals/ita/Belovs08,DBLP:journals/int/EndrullisHK11}. For
1CA reductions, it is possible to construct suprema when constrained to
ultimately periodic streams or sparse streams. However, as shown later in
\cref{thm:no_upper_bound}, upper bounds of pairs of streams do not exist in
general.

\begin{corollary}\label{thm:period_suprema}
    $\ultperiod_{\lcm(n, m)}$ is a supremum of $\{\ultperiod_n, \ultperiod_m\}$.
\end{corollary}
\begin{proof}
    $\ultperiod_{\lcm(n, m)}$ is an upper bound of $\ultperiod_n, \ultperiod_m$,
    by \cref{thm:periodic_divisibility}. To show that it is a supremum, let
    $X \in \hierarchy$ such that $X \cagt \ultperiod_n, \ultperiod_m$ and
    $\ultperiod_{\lcm(n, m)} \cagt X$. Then, we get $X \in \ultperiod$, as
    $\ultperiod$ is closed under $\cared$ by \cref{thm:closed_down}, and thus $X
    \caeq \ultperiod_{\lcm(n, m)}$ by \cref{thm:periodic_divisibility}.
    \qed
\end{proof}

Constructing suprema of sparse streams is a bit more involved, but still
possible. 

\begin{lemma}\label{thm:sparse_suprema}
    Let $\sigma, \tau$ be sparse. Then, $\{\sigma, \tau\}$ has a supremum.
\end{lemma}
\begin{proof}
Given sparse $\sigma$ and $\tau$ in $2^\omega$, we shall encode $\sigma$ and
$\tau$ into a new stream $\xi$ to form a supremum. This can be done by carefully
avoiding patterns that can be mixed up if padded with 0's left-and-right. The
following is one example of such an encoding:
\begin{itemize}
    \item If $\sigma(i) = 1$ and $\tau(i) = 1$, let $\xi(i, i+1, i+2, i+3) = 1000$,
    \item if $\sigma(i) = 1$ and $\tau(i+1) = 1$, let $\xi(i, i+1, i+2, i+3) = 1100$,
    \item if $\sigma(i) = 1$ and $\tau(i+2) = 1$, let $\xi(i, i+1, i+2, i+3) = 1010$,
    \item if $\sigma(i) = 1$ and $\tau(i+3) = 1$, let $\xi(i, i+1, i+2, i+3) = 1110$,
    \item if $\tau(i) = 1$ and $\sigma(i+1) = 1$, let $\xi(i, i+1, i+2, i+3) = 1001$,
    \item if $\tau(i) = 1$ and $\sigma(i+2) = 1$, let $\xi(i, i+1, i+2, i+3) = 1101$,
    \item if $\tau(i) = 1$ and $\sigma(i+3) = 1$, let $\xi(i, i+1, i+2, i+3) = 1011$,
    \item for any yet undefined index $i$, let $\xi(i) = 0$.
\end{itemize} 
From some point $n$ on, every 1 (in both $\sigma$ and $\tau$) is followed by at
least seven 0's, which ensures that the patterns above will not overlap.
Anything occurring before point $n$ can be ignored, as it can be addressed via a
finite mutation. Therefore, $\xi \cared \sigma, \tau$. To show that this upper
bound is also a supremum, consider an arbitrary upper bound $\xi'$ of both
$\sigma$ and $\tau$, and respective machines $M = (2, N, \delta)$ and $M' = (2,
N', \delta')$. Then, at a high-level, one can obtain a reduction from $\xi'$ to
$\xi$ straightforwardly, as $\xi(i)$ is defined with respect to $\sigma(i)$ and
$\tau(i)$ in the scheme above, and these are identified by $\delta(\neigh_{\xi',
N}(i))$ and $\delta'(\neigh_{\xi', N'}(i))$ respectively.
\qed
\end{proof}

Note that $\xi$ as constructed in \cref{thm:sparse_suprema} is not sparse (but
weakly sparse), so one cannot inductively apply the scheme above to obtain
suprema of finite sets of sparse streams. Nevertheless, for any finite set of
sparse streams, a similar argument can provide a suitable encoding, albeit a
more complicated one.

The construction of suprema in \cref{thm:period_suprema,thm:sparse_suprema} was
possible since, informally, not so much information is packed in them, allowing
us to compress two streams into a single one. In contrast, we next construct
streams that pack too much information to allow an upper bound, and in fact
belong to $\mathcal{S}$.

\begin{lemma}\label{thm:no_upper_bound}
    There exist $\sigma, \tau \in \abc^\omega$ without an upper bound.
\end{lemma}
\begin{proof}
Construct $\sigma$ and $\tau$ so that for every $k$ and every pair $w_1, w_2 \in
\Sigma^k$, block $w_1$ appears in $\sigma$ and is aligned with $w_2$ in $\tau$.
Suppose, for a contradiction, that $\xi$ reduces to $\sigma$ and $\tau$ via
machines $M_1 = (\Sigma, N_1, \delta_1)$ and $M_2 = (\Sigma, N_2, \delta_2)$,
respectively. Let $N = \max(N_1, N_2)$, and, extend $\delta_i : (\abc \cup
\{\bot\})^{2N+1} \to \abc$ with $i = 1, 2$, for convenience. Notice that a block
of $l$ letters in $\sigma$ (or $\tau$) depends only on a block of length $2N +
l$ in $\xi$. Moreover, such a block of $\xi$ needs to encode both $l$-sized
blocks of $\sigma$ and $\tau$, which is impossible for $l > 2N$ by construction.
More formally, we define maps of the form $g, h : \abc^{2N + l} \to \abc^l$ as
follows:
\begin{align*}
g(x_1, \dots, x_{2N+l}) &:= \prod_{j=1}^{l} \delta_1(x_j, \dots, x_{2N+j}),\\
h(x_1, \dots, x_{2N+l}) &:= \prod_{j=1}^{l} \delta_2(x_j, \dots, x_{2N+j}).
\end{align*}
Then, by letting $r(x) := g(x) \cdot h(x)$, it holds that 
\[
r(\xi(i-N), \dots, \xi(i+N+l-1)) = \sigma(i : i+l-1) \cdot \tau(i : i+l-1),
\]
given that $\xi \cared \sigma, \tau$ via $\delta_1$ and $\delta_2$. Function $r
: \abc^{2N + l} \to \abc^{l + l}$ must be surjective by construction of $\sigma$
and $\tau$, as they contain all pairs of subwords of length $l$. However, this
is not possible when $l > 2N$; yet another proof by diagonalization. 
\qed
\end{proof}

\section{An Algorithm for Determining Reducibility}
\label{sec.algorithm}

We can restate 1CA-reducibilty compactly in a nice slogan:

\begin{center}
\noindent\fbox{%
    \hfill
    $\sigma \cared \tau$ iff $\exists N \in \N, \forall i, i' \in \N,
    \neigh_{\sigma,N}(i) = \neigh_{\sigma,N}(i') \Rightarrow \tau(i) = \tau(i')$. 
    \hfill
}
\end{center}

The slogan suggests a brute-force algorithm for determining 1CA reducibility by
iterating over all neighborhood radii and checking for mismatches. We introduce
\cref{alg:ca_reduction}, which takes as input two streams $\sigma, \tau \in
\abc^\omega$, treated as oracles, along with a maximum query budget
$c_\text{max} \in \N$, which limits how many symbols of $\sigma$ and $\tau$ may
be accessed. It also uses a tunable hyperparameter $\alpha \in \mathbb{Q}^{>0}$
concerning the confidence required before returning ``Yes''.

The algorithm proceeds by iterating over increasing radius values $N = 0, 1, 2,
\dots$. For each $N$, it examines positions $i \geq N$ in the stream. At each
index $i$, it extracts the neighborhood $\neigh_{\sigma, N}(i)$ from $\sigma$. If
this neighborhood has not been seen before, it is added to a cache $V_N$ along
with the corresponding value $\tau(i)$. If it has been seen, the algorithm
verifies that $\tau(i)$ matches the previously stored value; a mismatch
indicates that the current radius $N$ is insufficient, prompting the algorithm
to break from the inner loop and increment $N$. Query usage is tracked
throughout, and once the query limit $c_\text{max}$ is reached, the algorithm
makes a judgment. If it has checked at least $\alpha |\abc|^{2N+1}$ distinct
indices for the current $N$, it returns ``Yes''; and otherwise returns ``No''.

\begin{algorithm}[tb]
\caption{Algorithm for approximating 1CA reducibility.}
\label{alg:ca_reduction}
\begin{algorithmic}[1]
\STATE \textbf{Input:} $\sigma, \tau \in \abc^\omega$, $c_\text{max} \in \N$, $\alpha \in \mathbb{Q}^{>0}$
\STATE \textbf{Output:} Estimate of whether $\sigma \cared \tau$
\STATE Initialize query counter $c = 0$
\FOR{$N = 0$ to $\infty$}
    \STATE Initialize cache $V_N = \{\}$
    \FOR{$i = N$ to $\infty$}
        \STATE $c \leftarrow c + 1$

        \IF{$c > c_\text{max}$}
            \IF{$i \geq \alpha |\abc|^{2N+1}$}
                \RETURN ``Yes''
            \ELSE
                \RETURN ``No''
            \ENDIF
        \ENDIF

        \IF{$\neigh_{\sigma,N}(i) \not \in V_N$}
            \STATE $V_N \leftarrow V_N \cup \{\neigh_{\sigma,N}(i) \mapsto \tau(i)\}$
        \ELSE
            \IF{$V_N(\neigh_{\sigma,N}(i)) \neq \tau(i)$}
                \STATE Break from the inner for loop.
            \ENDIF
        \ENDIF
    \ENDFOR
\ENDFOR
\end{algorithmic}
\end{algorithm}

Indeed, \cref{alg:ca_reduction} cannot decide whether a stream reduces to
another in finite time because there are infinitely many symbols to check. In
fact, if we encode streams as Turing machines, the decision problem is
unsurprisingly undecidable (see \cref{app.arithmetic}). Nevertheless, it has
``asymptotically correct'' behavior: $\sigma$ reduces to $\tau$ if and only if
the algorithm eventually always outputs ``Yes''. We formalize this in
\cref{thm:algo_correct}, where $P(\sigma, \tau, c_\text{max}, \alpha) \in
\{\text{``Yes''}$, $\text{``No''}\}$ denotes running the algorithm with
arguments $\sigma, \tau, c_\text{max}, \alpha$.

\begin{theorem}\label{thm:algo_correct} 
    For $\sigma, \tau \in \abc^\omega$, and $\alpha > 0$, $\sigma \cared \tau$
    if and only if $\exists c_0 \in \N$ where $\forall c_\text{max} \geq c_0$,
    $P(\sigma, \tau, c_\text{max}, \alpha) = \text{``Yes''}$.
\end{theorem}
\begin{proof}
    $(\Rightarrow)$ 
    Let $\sigma,\tau\in\abc^\omega$, and let $N$ be the smallest radius such
    that $\sigma \cared \tau$ with a 1CA of radius $N$. Then, $\forall N'< N,
    \exists i_{N'}, i'_{N'}$, where $\neigh_{\sigma, N'}(i_{N'}) = \neigh_{\sigma,
    N'}(i'_{N'})$ and $\tau(i_{N'}) \neq \tau(i'_{N'})$. By letting $c_0 =
    \sum_{N'<N} \max(i_{N'},i'_{N'})$, we have that \cref{alg:ca_reduction} will
    rule out all radii less than $N$ if it is passed a value for $c_\text{max}$
    greater than $c_0$. So, for any $\alpha \in \mathbb{R}^{>0}$, we can let
    $c_0' = c_0 + \alpha\lceil \abc^{2N+1}\rceil$, which means that for any
    $c_\text{max} > c_0'$, we get $P(\sigma, \tau, c_\text{max}, \alpha) =
    \text{``Yes''}$.

    $(\Leftarrow)$
    Let $\sigma, \tau \in \abc^\omega$, $\alpha > 0$, and assume that
    $\exists c_0 \in \N$ where $\forall c_\text{max} \geq c_0$, $P(\sigma,
    \tau, c_\text{max}, \alpha) = \text{``Yes''}$. Then, for some $N$, there do
    not exist $i, i' \in \N$ where $\neigh_{\sigma, N}(i) = \neigh_{\sigma, N}(i')$ and
    $\tau(i) \neq \tau(i')$, as the algorithm would otherwise eventually find
    this pair and output $\text{``No''}$. Therefore, $\sigma \cared \tau$. 
    \qed
\end{proof}

We reason about the negative case by contraposition in the following
corollary.

\begin{corollary}\label{thm:algo_correct3} For any $\sigma, \tau \in
    \abc^\omega$, and $\alpha > 0$, we have $\sigma \not\cared \tau$ 
    iff there are infinitely many $c_\text{max} \in \N$ where $P(\sigma,
    \tau, c_\text{max}, \alpha) = \text{``No''}$. \qed
\end{corollary}

However, \cref{thm:algo_correct3} does not imply that, if the reduction is not
possible, then the algorithm will eventually \emph{always} output ``No''. Of
course, one can construct two streams for which the algorithm outputs ``Yes''
infinitely many times (but not always) even though the answer should be ``No'',
by diagonalization. This is due to the issue of having to check infinitely many
symbols in finite time mentioned earlier. We construct such streams in
\cref{thm:algo_incorrect}.

\begin{lemma}\label{thm:algo_incorrect} $\exists \sigma, \tau \in
    \abc^\omega$, $\alpha > 0$, such that $\sigma \not \cared \tau$ even
    though there are infinitely many $c_\text{max} \in \N$ where $P(\sigma,
    \tau, c_\text{max}, \alpha) = \text{``Yes''}$.
\end{lemma}
\begin{proof}
    Let $\sigma = 0^\omega$ and $\tau := \prod_{i \in \N} 0^{\lceil \alpha
    |\abc|^{2i+1}\rceil} \cdot 1$. Then, $\sigma \not \cared \tau$ by
    \cref{thm:closed_down}, but there are infinitely many $c_\text{max} \geq
    c_0$ where $P(\sigma, \tau, c_\text{max}, \alpha) = \text{``Yes''}$.
    \qed
\end{proof}

In \cref{thm:algo_incorrect}, stream $\tau$ looks like $\sigma$ for
exponentially long (around $\alpha|\abc|^{2i+1}$ symbols), after which it looks
different to satisfy $\sigma \not \cared \tau$. In general, any counterexample
that makes the lemma work must also have these exponentially long parts where
the two streams look similar. The algorithm will thus eventually always yield
the correct output for any pair of streams without these exponential parts,
which is the case for many (if not most) streams of general interest. For
instance, ultimately periodic streams are eventually always correctly
classified.

\section{Alternative Definitions of Reducibility}\label{sec:alternatives}
\label{sec.alternatives}
This section considers extensions to the notion of 1CA reducibility introduced
so far. We provide preliminary results for hybrid and finite-word cellular
automata, and also briefly introduce alternatives. These may not be cellular
automata in the strict sense, but they are nevertheless cellular-automaton-like.

\subsection{Hybrid Cellular Automata}\label{sec:hybrid}

Some approaches to cryptography involve the use of hybrid cellular automata,
where multiple local rules are combined in an alternating manner to balance
various metrics of cryptographic
robustness~\cite{DBLP:conf/acri/ChakrabortyC12}. In the next definitions, we
introduce hybrid cellular automaton reducibility and degrees, and prove some
preliminary properties.

\begin{definition}
    A \emph{hybrid 1CA} is a tuple $M = (\abc, N, \{\delta_1, \dots,
    \delta_K\})$ defined like standard 1CAs, except for the global rule
    $f_M(\sigma)(i) = \delta_k (\neigh_{\sigma,N}(i))$ with $k \equiv i \pmod{K}$.
\end{definition}

\begin{definition}
    A \emph{hybrid 1CA reduction} is relation $\cared_h \subseteq \abc^\omega
    \times \abc^\omega$ where $\sigma \cared_h \tau$ if and only if there exists
    a h1CA $M$ such that $f_M(\sigma) = \tau$. Additionally, we let $\sigma
    \caredl_h \tau \Leftrightarrow \tau \cared_h \sigma$, and $\caeq_h =
    \cared_h \cap \caredl_h$. 
\end{definition}

\begin{definition}
    A \emph{$\caeq_h$-degree} is a member of the quotient set $\hierarchy_h =
    \abc^\omega \slash \caeq$, and we refer to $\hierarchy_h$ as the \emph{set
    of $\caeq_h$-degrees}.
    The relation $\caleq_h$ is defined for all $X, Y \in \hierarchy_h$ as
    $Y \caleq_h X \Leftrightarrow \forall \sigma \in X, \forall \tau \in Y, \ 
    \sigma \cared_h \tau$.
\end{definition}

\begin{proposition}\label{thm:hca_comp}
    Relation $\cared_h$ compares as follows with FST- and 1CA-reductions:
    \begin{enumerate}
        \item $\cared \subsetneq \cared_h$,
        \item $\cared_h \subsetneq \fstred$.
    \end{enumerate}
    This implies that $\cared_h$ is closed under finite mutations. 
\end{proposition}
\begin{proof}\quad
    \begin{enumerate}
        \item We have $\cared \subseteq \cared_h$ as any 1CA $M = (\sigma, N,
        \delta)$ has an equivalent h1CA $M' = (\sigma, N, \{\delta\})$, whereas
        $\cared_h \not \subseteq \cared$ because $0^\omega \not \cared (01)^\omega$
        but $0^\omega \cared_h (01)^\omega$ through $M'' = (\abc, 0, \{\delta_1,
        \delta_2\})$ where $\delta_1(\_) = 0$ and $\delta_2(\_) = 1$.
        \item This can be done with an FST as in \cref{thm:ca_subseteq_fst},
        with the difference that it must simulate a different $\delta_k$ 
        each step, possible by letting 
        $Q = (\abc \cup \{\bot\})^{2N+1} \times \{1,\dots,K\}$. \qed
    \end{enumerate}
\end{proof}

\begin{proposition}\label{hyca_pre}
    $(\abc^\omega, \cared_h)$ is a preorder. 
\end{proposition}
\begin{proof}
    Reflexivity is immediate. Take any $\sigma, \tau, \rho \in \abc^\omega$ and
    suppose that $\sigma \cared_h \tau$ and $\tau \cared_h \rho$ via $M_i =
    (\abc, N_i, \{\delta^{(i)}_1, \dots, \delta^{(i)}_{K_i}\})$ with $i = 1, 2$.
    Then, there exists 
    \[
    M_3 = (\abc, N_1 + N_2, \{\delta^{(3)}_1, \dots, \delta^{(3)}_{K_1 \cdot K_2}\})
    \]
    such that $\sigma \cared_h \rho$, which can be constructed similarly to
    \cref{thm:preorder}. Using the shorthand notation $\delta^{(3)}_{i, j} :=
    \delta^{(3)}_{(i-1) \cdot K_2 + j}$, let
        \[
        \trans^{(3)}_{i, j}(x_0, \dots, x_{2(N_1 + N_2)}) := \trans^{(2)}_j(\trans^{(1)}_i(v_0), \dots,\trans(v_{2 N_2})),
        \]
        where $v_k := (x_{k}, \dots, x_{k+2N_1})$. \qed
\end{proof}

\begin{lemma}\label{hyca_low}
    $\ultperiod$ is $\cared_h$-minimal for $\abc^\omega$, i.e., $\forall \sigma
    \in \abc^\omega$, $\forall \tau \in \ultperiod$, $\sigma \cared_h \tau$.
\end{lemma}
\begin{proof}
    Let $\sigma \in \abc^\omega$ and $\tau = x \cdot y^\omega \in \ultperiod$,
    with $x, y \in \abc^*$. Then, $\sigma \cared_h y^\omega$ with $M = (\abc,
    N, \{\trans_1, \dots, \trans_K\})$ where $N = 0$, $K = |y|$, and
    $\delta_i(\_) = y(i)$. We get $y^\omega \cared_h \tau$ since the relation
    is closed under finite mutations, and $\sigma \cared_h \tau$ by
    transitivity.
    \qed
\end{proof}

Since $\cared_h \subseteq \fstred$ by \cref{thm:hca_comp}, $\ultperiod$ must be
closed under $\caeq_h$, and thus also a degree. Therefore, $\ultperiod$ is the
minimal $\cared_h$-degree due to \cref{hyca_low}.

\begin{corollary}
   $\ultperiod$ is the minimal $\caeq_h$-degree. \qed
\end{corollary}

In conclusion, the hierarchy of streams arising from hybrid 1CAs has interesting
algebraic properties, and their presence in practical applications such as in
cryptography provides additional motivation to further explore this notion.

\subsection{Cellular Automata with Finite Word Output}\label{sec:finite_word}

We consider 1CAs where cells can output multiple letters, which can be
interpreted as having each cell split into multiple cells.

\begin{definition}
    A \emph{finite-word 1CA} is a triple $(\abc, N, \trans)$ defined like
    standard 1CAs, except that $\trans : (\abc \cup \{\bot\})^{2N+1} \to
    \abc^*$, and the global rule $f_M : \abc^\omega \rightharpoonup \abc^\omega$
    is $f_M(\sigma) = \prod_{i\in\N} \delta(\neigh_{\sigma,N}(i))$.

\end{definition}

\begin{definition}
    A \emph{finite-word 1CA reduction\/} is a relation $\cared_* \subseteq
    \abc^\omega \times \abc^\omega$ where $\sigma \cared_* \tau$ if and only if
    there exists a f1CA $M$ such that $f_M(\sigma) = \tau$. Additionally, we let
    $\sigma \caredl_* \tau \Leftrightarrow \tau \cared_* \sigma$, and $\caeq_* =
    \cared_* \cap \caredl_*$.
\end{definition}

\begin{proposition}\label{thm:fca_comp}
    $\cared_*$ compares as follows with the other notions:
    \begin{enumerate}
        \item $\cared \subsetneq \cared_*$,
        \item $\cared_* \subsetneq \fstred$,
        \item $\cared_h \not \subseteq \cared_*$, and $\cared_* \not \subseteq \cared_h$.
    \end{enumerate}
\end{proposition}
\begin{proof}
    1.\ and 2.\ are similar to the proof in \cref{thm:hca_comp}.

    To show $\cared_* \not \subseteq \cared_h$, take any $\xi \in \abc^\omega$
    and let $\xi \cared_* f_M(\xi)$ with the finite-word 1CA $M = (\abc, 0,
    \delta)$ and $\delta(x) = xx$. Hybrid 1CAs cannot perform this reduction in
    general, say, when $\xi$ is noncomputable. Therefore, we conclude $\cared_h
    \subsetneq \cared_*$.
    To show $\cared_h \not \subseteq \cared_*$, define
    \[
    \sigma := \prod_{i=0}^\infty 0^{2^i} \cdot 1 \quad \text{and} \quad \tau := \prod_{i=0}^\infty 0^{3 \cdot 2^{2i} + 1} \cdot 1.
    \] 
    We get $\sigma \cared_h \tau$ by replacing every letter 1 at an even index
    in $\sigma$ with a 0 via the h1CA $\delta_1(x) = x$ and $\delta_2(\_) = 0$,
    as $3 \cdot 2^{2i} + 1 = 2^{2i} + 1 + 2^{2i+1}$. In contrast, $\sigma \not
    \cared_* \tau$ given that, for any $N$, local neighborhoods in $\sigma$ of
    the form $0^N \cdot 1 \cdot 0^N$ must map to both 0 and 1 in $\tau$.
    \qed
\end{proof}

As a corollary of \cref{thm:mutations} and \cref{thm:fca_comp}, $\cared_*$ is
closed under finite mutations.

\begin{lemma}
   $\ultperiod$ is $\cared_*$-minimal for $\abc^\omega$, that is \[\forall \sigma \in
   \abc^\omega, \forall \tau \in \ultperiod, \sigma \cared_* \tau.\]
\end{lemma}
\begin{proof}
   Let $\sigma \in \abc^\omega$ and $\tau = x \cdot y^\omega \in \ultperiod$,
   such that $x, y \in \abc^\omega$. Then, we get $\sigma \cared_* \tau$ via the
   1CA with rule
   \[
       \trans(c, \_, \_) = \begin{cases}
           x \quad &\text{ if } \quad c = \bot\\
           y \quad &\text{ else.} 
       \end{cases}
       \qed
   \]
\end{proof}

However, the reduction is not transitive as we show next:

\begin{proposition}\label{thm:f1ca_not_transitive}
    Relation $\cared_*$ is not transitive on $\abc^\omega$.
\end{proposition}
\begin{proof}
    Fix $\abc = 3 = \{0, 1, 2\}$ and let $\xi \in \{0, 1\}^\omega$ be arbitrary.
    Then, define stream
    \[
    \rho = \prod_{i \in \N} \xi(i)\cdot 2^i = 
    \xi(0)\cdot\xi(1)\cdot 2 \cdot\xi(2)\cdot22\cdot\xi(3)\cdot222\cdot \xi(4)\cdot 2222\cdot\dots,
    \]
    and let $M$ be the 1CA that computes pairwise XOR (see \cref{ex:1ca_xor}):
    $f_M(x\app y \app\sigma) = (x \oplus y) \cdot f_M(y \app\sigma)$. We can reduce $\rho
    \cared_* \xi$ with $\trans(x) = \varepsilon$ if $x = 2$ and $\trans(x) = x$
    otherwise. Also $\xi \cared_* f_M(\xi)$, whereas $\rho \not \cared_*
    f_M(\xi)$ for most $\xi \in \abc^\omega$, as the gaps between the $\xi(i)$
    increase.
    \qed
\end{proof}

Since $(\abc^\omega, \cared_*)$ is not a preorder, this model is less
interesting than our 1CA reductions. Still, future work could study the
transitive closure of $\cared_*$ as a measure of stream complexity.

\subsection{Other Notions of Reducibility}\label{sec:other_alternatives}

\paragraph{Limit-Reducibility}
Allowing indefinite execution of CA rules leads to \emph{limit
sets}~\cite{wolfram1984computation}, motivating \emph{limit 1CA} reducibility,
where streams are transformed through CA whose cells may run indefinitely. Since
elementary cellular automata are Turing
complete~\cite{DBLP:journals/compsys/000104}, limit 1CA reductions may share
similarities with Turing reductions.

\paragraph{Two-Sided Cellular Automata}
FST reducibility naturally extends to two-sided
streams~\cite{bosma2017ordering}, and can be done similarly with 1CA. Unlike
one-sided 1CAs, two-sided 1CAs are not closed under finite mutations and can
have finite degrees (e.g., the degree of $\dots0000\dots$ is finite), even
though results on periodic and sparse streams are likely to carry over.

\section{Conclusions and Future Work}

To summarize, $(\hierarchy, \caleq)$ is a poset where reductions are invariant
under finite changes to the streams. We fully classify ultimately periodic
degrees, by identifying that $\ultperiod_n \caleq \ultperiod_m$ if and only if
$n \mid m$. Additionally, we found infinite descending chains by strategically
compressing ultimately periodic streams into a single stream, showing that the
hierarchy is not well-founded. 
Sparse streams are generally atoms, their degrees only contain weakly sparse
streams, and weakly sparse streams can encode periodic streams within them.
Maximal streams must include every finite word infinitely many times, and
suprema of sets of streams do not generally exist. \Cref{fig:hierarchy} presents
a schematic picture of the hierarchy.

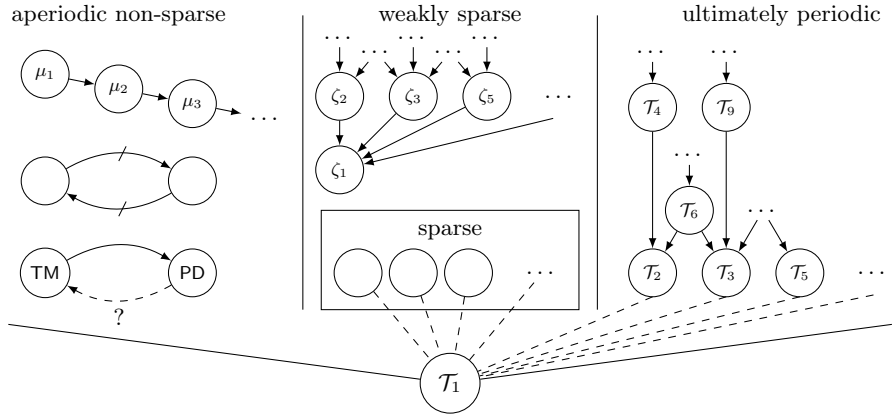
\begin{figure}[tb]
\centering
\resizebox{\columnwidth}{!}{
    {\small
      \begin{tikzpicture}[->, >=latex]
    \node[state] at (0, 0) (t1) {$\ultperiod_1$};
    \node[] at (-4.5, 5.0) {aperiodic non-sparse};
    \node[] at (0, 5.0) {weakly sparse};
    \node[] at (4.5, 5.0) {ultimately periodic};

    \draw[-] (t1) -- (6, 0.8);
    \draw[-] (t1) -- (-6, 0.8);
    \draw[-] (-2, 1) -- (-2, 5);
    \draw[-] (2, 1) -- (2, 5);

    \node[state, scale=0.8] at (-5.5, 4.2) (m0) {$\mu_1$};
    \node[state, scale=0.8] at (-4.5, 4.0) (m1) {$\mu_2$};
    \node[state, scale=0.8] at (-3.5, 3.8) (m2) {$\mu_3$};
    \node[] at (-2.5, 3.6) (mm) {$\dots$};

    \draw[->] (m0) -- (m1);
    \draw[->] (m1) -- (m2);
    \draw[->] (m2) -- (mm);

    \node[state, scale=0.8] at (-5.5, 2.75) (i0) {~};
    \node[state, scale=0.8] at (-3.5, 2.75) (i1) {~};
    \draw (i0) edge[->, bend left] (i1);
    \draw (i1) edge[->, bend left] (i0);
    \draw[-] (-4.5, 3.02) -- (-4.4, 3.22);
    \draw[-] (-4.5, 2.27) -- (-4.4, 2.47);

    \node[state, scale=0.8] at (-5.5, 1.5) (tm) {$\TM$};
    \node[state, scale=0.8] at (-3.5, 1.5) (pd) {$\PD$};
    \draw (tm) edge[->, bend left] (pd);
    \draw[dashed] (pd) edge[->, bend left] node[below] {?} (tm);

    \node[state, scale=0.8] at (-1.5, 2.9) (z1) {$\zeta_1$};
    \node[state, scale=0.8] at (-1.5, 3.9) (z2) {$\zeta_2$};
    \node[state, scale=0.8] at (-0.5, 3.9) (z3) {$\zeta_3$};
    \node[state, scale=0.8] at (0.5, 3.9) (z5) {$\zeta_5$};

    \node[] at (-1.5, 4.7) (dz2) {$\dots$};
    \node[] at (-0.5, 4.7) (dz3) {$\dots$};
    \node[] at (0.5, 4.7) (dz5) {$\dots$};
    \node[] at (1.5, 3.9) (dz7) {$\dots$};
    \node[] at (-1.0, 4.5) (dz23) {$\dots$};
    \node[] at (0.0, 4.5) (dz35) {$\dots$};

    \draw[->] (z2) -- (z1);
    \draw[->] (z3) -- (z1);
    \draw[->] (z5) -- (z1);
    \draw[->] (1.4, 3.6) -- (z1);
    \draw[->] (dz2) -- (z2);
    \draw[->] (dz3) -- (z3);
    \draw[->] (dz5) -- (z5);
    \draw[->] (dz23) -- (z2);
    \draw[->] (dz23) -- (z3);
    \draw[->] (dz35) -- (z3);
    \draw[->] (dz35) -- (z5);

    \draw[] (-1.75,2.35) rectangle (1.75,1.0);
    \node[] at (0, 2.05) {sparse};
    \node[state, scale=0.8] (s1) at (-1.25, 1.5) {};
    \node[state, scale=0.8] (s2) at (-0.5, 1.5) {};
    \node[state, scale=0.8] (s3) at (0.25, 1.5) {};
    \node[] at (1.25, 1.5) (ds) {$\dots$};

    \draw[-, dashed] (s1) -- (t1);
    \draw[-, dashed] (s2) -- (t1);
    \draw[-, dashed] (s3) -- (t1);
    \draw[-, dashed] (ds) -- (t1);

    \node[state, scale=0.8] at (2.75, 1.5) (t2) {$\ultperiod_2$};
    \node[state, scale=0.8] at (3.75, 1.5) (t3) {$\ultperiod_3$};
    \node[state, scale=0.8] at (4.75, 1.5) (t5) {$\ultperiod_5$};
    \node[state, scale=0.8] at (2.75, 3.75) (t4) {$\ultperiod_4$};
    \node[state, scale=0.8] at (3.25, 2.35) (t6) {$\ultperiod_6$};
    \node[state, scale=0.8] at (3.75, 3.75) (t9) {$\ultperiod_9$};

    \node[] at (5.75, 1.5) (dt7) {$\dots$};
    \node[] at (2.75, 4.5) (dt4) {$\dots$};
    \node[] at (3.75, 4.5) (dt9) {$\dots$};
    \node[] at (3.25, 3.1) (dt6) {$\dots$};
    \node[] at (4.25, 2.35) (dt15) {$\dots$};

    \draw[-, dashed] (2.75, 1.175) -- (t1);
    \draw[-, dashed] (3.75, 1.175) -- (t1);
    \draw[-, dashed] (4.75, 1.175) -- (t1);
    \draw[-, dashed] (5.75, 1.2) -- (t1);

    \draw[->] (t4) -- (t2);
    \draw[->] (t9) -- (t3);
    \draw[->] (t6) -- (t2);
    \draw[->] (t6) -- (t3);
    \draw[->] (dt4) -- (t4);
    \draw[->] (dt9) -- (t9);
    \draw[->] (dt6) -- (t6);
    \draw[->] (dt15) -- (t3);
    \draw[->] (dt15) -- (t5);
      \end{tikzpicture}
      }
}
\caption{An illustration of the hierarchy of 1CA degrees. Each circle represents a degree, while arrows represent reducibility. It includes infinite descending chains, incomparable degrees, $\TM \cared \PD$, sparse and weakly sparse degrees, and ultimately periodic degrees.}
\label{fig:hierarchy}
\end{figure}

\paragraph{Future work} As discussed in \cref{sec:maxima}, the existence or
non-existence of maximal degrees is an open question. Additionally, we have left
out the classification of particular, well-known streams, as we consider it
orthogonal to our work. We conjecture $\TM \not \cared \PD$, which may be
proven by showing that the XOR operation cannot be undone; similarly to the
proof in \cref{thm:non_maximal}. Lastly, \cref{sec:alternatives} introduces
alternative notions of reducibility, which could serve as potential directions
for future work.

\begin{credits}
\subsubsection{\ackname}
Partially funded by the ERC Starting Grant 101077178 (DEUCE).

\subsubsection{\discintname}
The authors have no competing interests to declare that are relevant to the
content of this article.
\end{credits}

\bibliographystyle{splncs04}
\bibliography{bibliography}

\appendix

\section{Proofs of the classifications in the arithmetic hierarchy for \Cref{sec.algorithm}}\label{app.arithmetic}

In \cref{sec.algorithm}, we introduce an algorithm that takes two streams
as inputs and estimates whether one reduces to the other. Throughout, the input
streams are treated as oracles, without specifying how they are encoded. If we
represent streams via Turing machines, we can formalize a specific decision
problem and place it in the arithmetic hierarchy. Given the hardness of
reasoning about Turing machines, the problem is indeed undecidable in this
setting.

First, we assume (as usual) an effective enumeration of the partial computable
functions $\{\varphi_i\}_{i\in\N}$. We then say that a stream $\sigma \in
\abc^\omega$ is \emph{computable} whenever $\exists i \in \N, \forall j \in
\N, \varphi_i(j) = \sigma(j)$. This allows posing cellular automaton
reducibility as an arithmetical property:
\begin{definition}\label{def:computable_red}
    \emph{Computable 1CA reducibility} is the set:
    \begin{align*}
    \cacomp = \{\langle s, t \rangle \in \N~\mid~&\varphi_s, \varphi_t \text{ are total and} exists N \! \in \N, \forall i, i' \! \in \N,\\
    &\neigh_{\varphi_s,N}(i) = \neigh_{\varphi_s,N}(i') \Rightarrow \varphi_t(i) = \varphi_t(i')\}.
    \end{align*}
\end{definition}

\begin{lemma}
    $\cacomp \in \arithd^0_3$.
\end{lemma}
\begin{proof}
This is a corollary of \cref{def:computable_red}.
\qed
\end{proof}

\begin{lemma}
    $\cacomp\in\ariths^0_2$-Hard.
\end{lemma}
\begin{proof}
    It is enough to show that the set $\text{Fin} = \{e \in \N \mid
    \text{dom}(\varphi_e) \text{ is finite}\}$, a well-known
    $\ariths^0_2$-complete problem~\cite{soare1987}, many-one reduces to $\cacomp$.
    First, we introduce the next function:
    \[
    \varphi_e(k)_n = 
    \begin{cases}
        0 \quad \text{ if } \varphi_e(k) \text{ has halted precisely at step } n,\\
        1 \quad \text{ otherwise.}
    \end{cases}
    \]
    Let $\psi_e : \N \to \{0, 1\}$ be the total computable function obtained
    by dovetailing through all $\varphi_e(k)_n$. That is, given a bijection $\pi
    : \mathbb{N} \to \mathbb{N}^2$, we let $\psi_e(i) =
    \varphi_e(x)_y$ with $\pi(i) = (x, y)$. Then, we have $\varphi_e \in \text{Fin}
    \Leftrightarrow 1^\omega \cacomp \psi_e$, meaning that $\text{Fin} \leq_m \cacomp$,
    so $\cacomp$ is $\ariths^0_2$-Hard.
    \qed
\end{proof}

\begin{lemma}
    $\cacomp\in\arithp^0_2$-Hard.
\end{lemma}
\begin{proof}
    The set $\text{Tot} = \{e \in \N \mid \varphi_e \text{ is total}\}$ is a
    well-known $\arithp^0_2$-complete problem~\cite{soare1987}, and $\text{Tot}$ many-one reduces
    to $\cacomp$ via $f(n) = \langle n, n\rangle$.
    \qed
\end{proof}

\section{$\orbit$-Orbits and Elementary Cellular Automata}\label{app:Orbits}

The \emph{first difference operator} is defined as
$
    \orbit(x:y:\sigma) = (x \oplus y) : \orbit(y : \sigma),
$
where $\oplus$ denotes addition modulo 2 (XOR). It is well known that
$\orbit(\text{TM}) = \text{PD}$, as mentioned in
\cref{ex:1ca_xor}. Motivated by this result,
\cite{DBLP:journals/delta-orbits} ask an essential question regarding the first
difference operator: what do we encounter by iterating $\orbit$ over and over?
To answer this question, they introduce the following concept: Given a stream
$\sigma$, its \textit{$\orbit$-orbit} is the sequence 
$
    \sigma, \orbit(\sigma), \orbit^2(\sigma), \orbit^3(\sigma), \dots
$
$\orbit$-orbits enable the ``graphical'' analysis of streams, where each stream
in the sequence of orbits is displayed as a row of pixels in an image. That is,
the first row of pixels in the image corresponds to $\sigma$, the second row to
$\orbit(\sigma)$, and so on. By employing this technique of ``graphical''
evaluation,~\cite{DBLP:journals/delta-orbits} visualize the $\orbit$-orbits of
the Sierpiński sequence, $\mathsf{S} = 00111100011000011\dots$, and the Mephisto
Waltz sequence, $\mathsf{W} = 0010011100010011\dots$, to make the curious
discovery that $\orbit^2(\mathsf{S}) = \orbit^3(\mathsf{W})$. These $\orbit$-orbits
are shown in \cref{fig:orbit}. Since the definitions of these two
sequences are seemingly unrelated to each other, this striking fact would be
hard to notice without the use of $\orbit$-orbits. 

As discussed in \cref{ex:1ca_xor}, 1CA can perform function $\orbit$.
This cellular automaton is an elementary cellular automaton (ECA), namely, Rule
102, which is classified in Class 3 of ECA together with Rule~30
\cite{wolfram1983statistical}. Moreover, a $\orbit$-orbit of a sequence is
precisely its space-time diagram under Rule~102. Interestingly, natural
structures and processes are often compared with the behavior of ECA in Class
3~\cite{coombes2009geometry}.
A natural next step would involve the analysis of the orbits produced by other
ECA, to look for connections between other well-known sequences. One may even
automate the process with an algorithm similar to
\cref{alg:ca_reduction}.

\begin{figure}[tb]
    \begin{center}
    \includegraphics[width=0.45\textwidth]{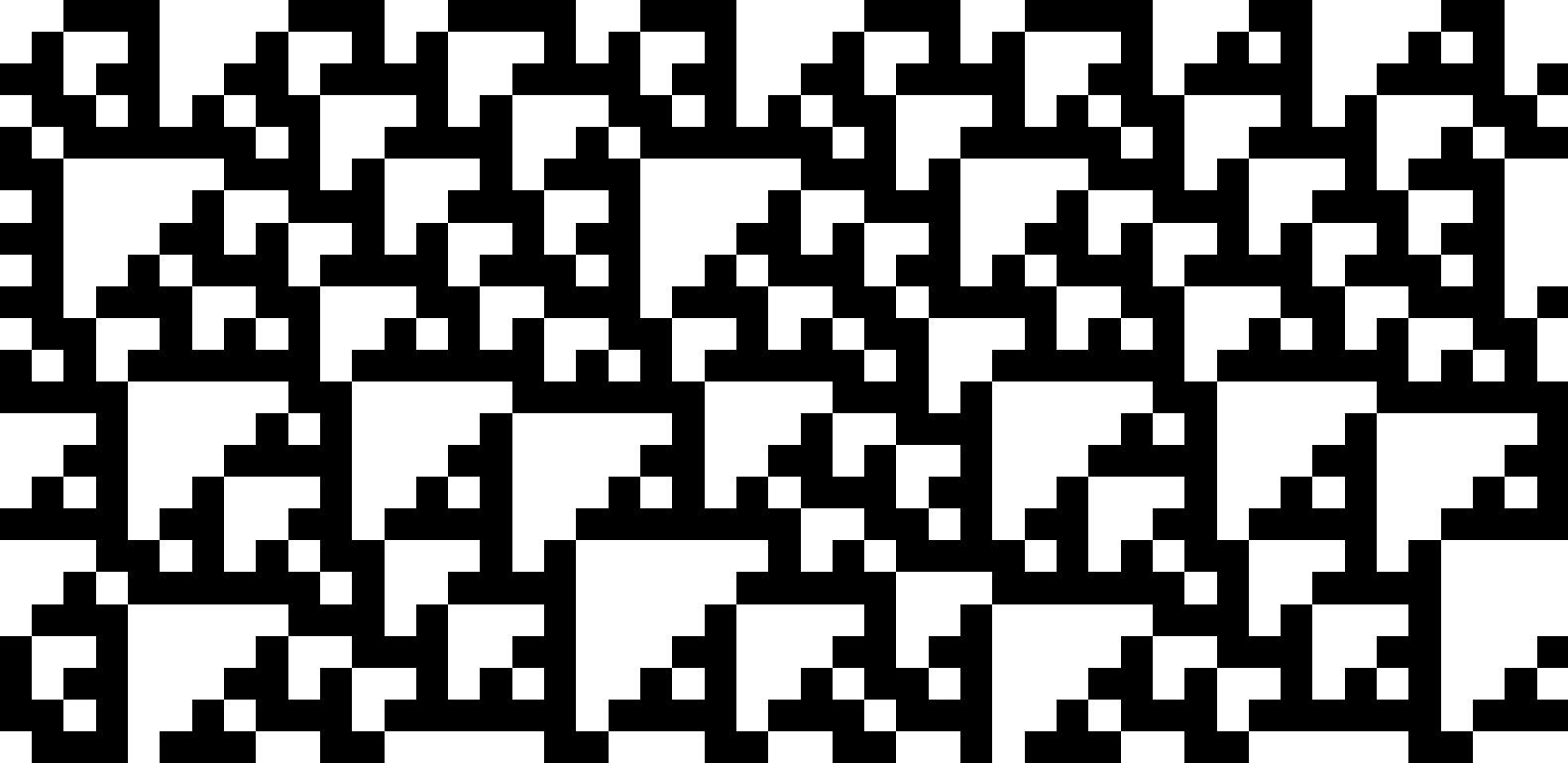}
    \hfill
    \includegraphics[width=0.45\textwidth]{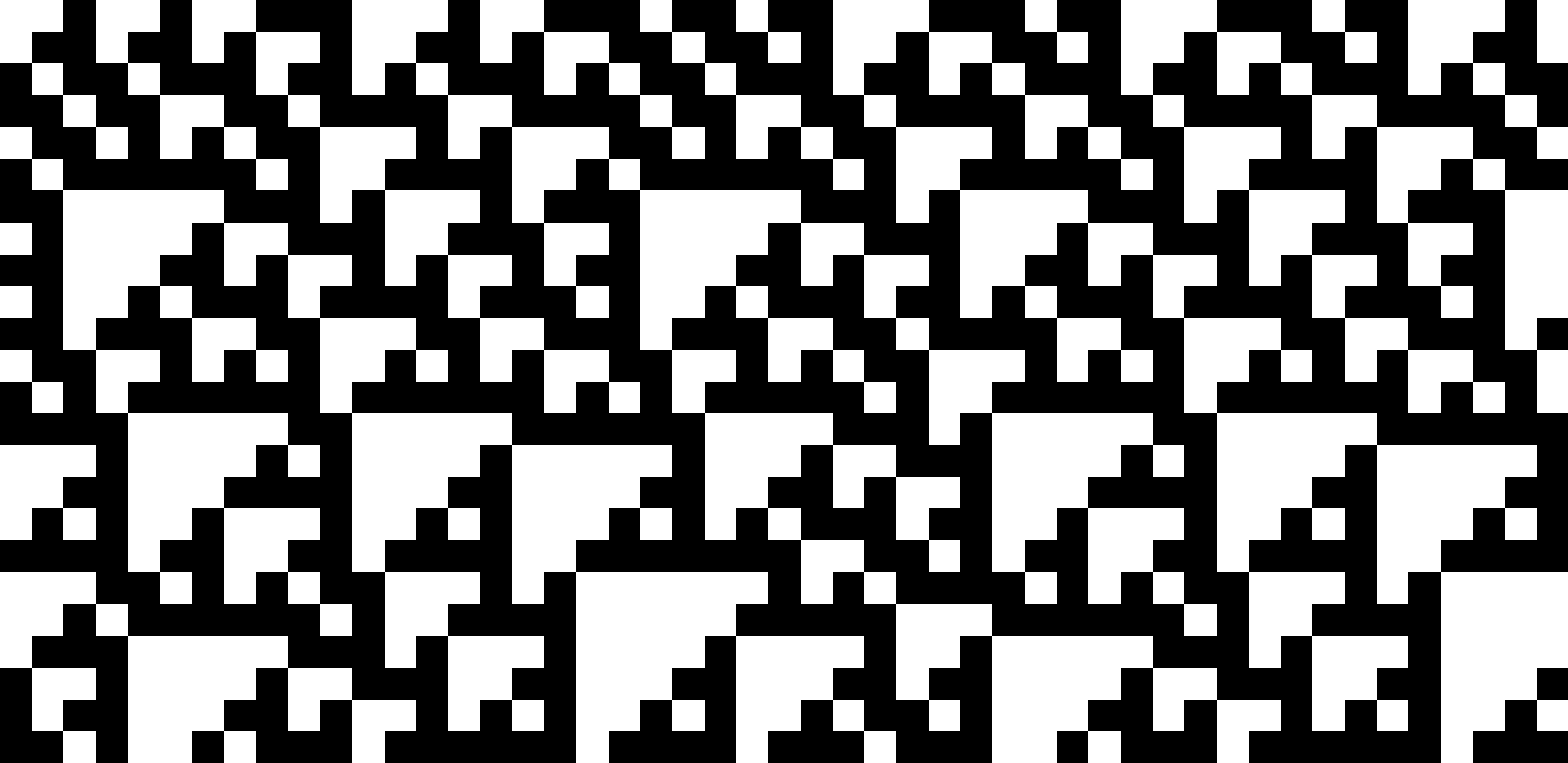}
    \end{center}
    \caption{The $\orbit$-orbits of the Sierpiński (left) and Mephisto-Waltz (right) sequences.}
    \label{fig:orbit}
\end{figure}

\end{document}